%% file: main.tex
\documentclass[11pt]{article}

\input{preamble}
\hypersetup{
	pdftitle={The planted tensor problem and its cryptographic applications},
	pdfauthor={Anonymous Author(s)},
	bookmarksnumbered=true,     
	bookmarksopen=true,         
	bookmarksopenlevel=1,       
	colorlinks=true,            
	pdfstartview=Fit,           
	pdfpagemode=UseOutlines,
	pdfpagelayout=OneColumn,
	pdfstartview=FitH
}

\title{The planted tensor problem over finite fields: \\
algorithms and cryptography}
\author{
Yuxuan Liu\thanks{State Key Lab of Processors, Institute of Computing Technology, Chinese Academy of Sciences, Beijing 100190, China\\
School of Computer Science and Technology, University of Chinese Academy of Sciences, Beijing 100049,
China}
\and 
Youming Qiao\thanks{University of Technology Sydney}
\and 
Gang Tang\thanks{Wuhan University}
\and 
Chuanqi Zhang\thanks{Monash University}
}

\date{}

\begin{document}

\maketitle

\pagenumbering{gobble}
\begin{abstract}
Inspired by the planted clique problem for random graphs, we introduce the planted totally-isotropic space problem for random tensors as follows. Let $U\cong \F_q^n$ and $W\cong \F_q^m$ be finite-dimensional vector spaces over a finite field $\F_q$. Given $d\in \N$, choose a random \(d\)-dimensional subspace \(V\leq U\), and construct a random alternating bilinear map $\phi:U\times U\to W$ subject to the constraint \(\phi(V,V)=0\). Such a $V$ is known as a totally-isotropic space of $\phi$, and the goal is to recover $V$. 

Building on the recent probabilistic analysis of random tensors (Pham--Qiao--Wigderson--Wigderson, \emph{in progress}), we initiate the study of the algorithmic hardness of this problem. 
Setting $m=\lceil n/\log n\rceil$, we show that this problem admits an average-case polynomial-time algorithm for $d\geq n/2$, by leveraging recent advances on the non-commutative rank problem.
We also show that this problem admits a $q^{O(n\log n)}$-time algorithm. We carry out algorithmic experiments using polynomial-system solving. From these results, we conjecture that the planted totally-isotropic space problem for $d=\lceil n/C\rceil$ with some constant $C\geq 3$ is exponentially hard.  

Based on this evidence of computational hardness, we explore cryptographic applications of the planted totally-isotropic space problem and related planted tensor problems. We present private simultaneous messages and secret sharing protocols based on planted tensor problems, following the protocols based on planted subgraphs in (Abram--Beimel--Ishai--Kushilevitz--Narayanan, \emph{TCC}'23). At the same security level, the public information size of protocols based on planted subgraphs is (moderately) exponential in that of protocols based on planted tensors, while the communication costs of these protocols are polynomially related. 
\end{abstract}

\newpage
\tableofcontents
\newpage

\pagenumbering{arabic}
\setcounter{page}{1}

\section{Introduction}

The \textit{planted clique problem}, introduced by Jerrum \cite{Jer92} and Ku\v{c}era \cite{Ku95}, is a flagship problem in average-case complexity. The problem states that it is computationally hard
to find a clique of size $d$ that has been ``planted'' in a random graph of order $n$ as long as $d$ is not too large. It has been extensively studied from algorithmic, statistical, and complexity-theoretic perspectives \cite{AKS98,FK03,Ros08,Ros10,FRG13,BHK16,ABDr18,MRS21}, serving as a canonical example of a problem with a conjectured statistical-computational gap.

Cryptographic protocols based on the assumed hardness of the planted clique problem have been devised, including one-way functions, zero-knowledge protocols \cite{JP00}, and public-key encryption schemes \cite{ABW10,BKR23}. Recent new cryptographic and security applications based on planted cliques include planting undetectable backdoors in machine learning models \cite{GKVZ22}, private simultaneous messages and secret sharing \cite{ABIKN23}, as well as the public-key encryption scheme \cite{GHJS25}. In particular, \cite{GHJS25} relies on the standard planted clique hardness assumption instead of the variant used in \cite{ABW10}.

With the increasing applications of the planted clique hardness assumption in cryptography, there is an inconvenient issue in putting these protocols into action. That is, the planted clique problem in the standard model can be solved in quasipolynomial time. As a result, the computational hardness assumption relies on the \emph{quasipolynomial-time} hardness against \emph{polynomial-time} algorithms, whereas most practical cryptographic protocols rely on the \emph{(sub)exponential-time} hardness against \emph{polynomial-time} algorithms. 

Based on the above, the main motivation of this paper is the following question:
\begin{quote}
    \textit{Is there a structure with a planted substructure paradigm that can be potentially exponentially hard and useful for cryptographic applications?}
\end{quote}
Our investigation suggests that \emph{tensors} could serve as an answer to the above question. In the following, we introduce a tensor version of the planted clique problem, examine algorithms for this problem, and explore its cryptographic applications. 

\subsection{From planted cliques to planted totally-isotropic spaces}

\paragraph{Planting cliques, planting independent sets.} We shall focus on the planted clique problem in the Erd\H{o}s--R\'enyi model $\ER(n, p=1/2)$. In this case, we can plant independent sets which is the same as planting cliques and then taking the complement graph. In particular, all algorithms and complexity-theoretic results for planted clique can be transferred to planted independent set in a straightforward fashion.

\paragraph{Graphs and tensors.} A tensor is a multilinear map. In particular, an alternating bilinear map may be viewed as an order-3 tensor. A classical correspondence between graphs and tensors goes back to Tutte \cite{Tut47}, Edmonds \cite{Edmonds67}, and Lov\'asz \cite{Lov79}, who showed a connection between perfect matchings and full-rank matrices. 
Recently, more connections between graphs and bilinear maps were discovered \cite{LQ17,BCGQS21,LQWWZ23,LQWWZ25}. 

We explain the construction of a bilinear map from a graph. Let $G=([n], E)$ be a graph with vertex set $[n]=\{1, \dots, n\}$ and edge set $E\subseteq \binom{[n]}{2}$, $|E|=m$. Given $\{i, j\}\subseteq [n]$, $i<j$, the elementary alternating matrix corresponding to $\{i,j\}$, $A_{i,j}$, is the $n\times n$ matrix with the $(i,j)$th entry $1$, $(j, i)$th entry $-1$, and all other entries equal to $0$. Let $U$ and $W$ be two vector spaces. A bilinear map $\phi:U\times U\to W$ is \emph{alternating} if for any $u\in U$, $\phi(u, u)=0$. We adopt the convention that $\F^n$ consists of column vectors.

From $G=([n], E)$, $|E|=m$, we first order the edges arbitrarily as $(e_1, \dots, e_m)$, where $e_k=\{i_k, j_k\}$. Then we can construct a tuple of elementary alternating matrices $(A_1, \dots, A_m)$ where $A_k$ is the elementary alternating matrix corresponding to $e_k$. This alternating matrix tuple then yields an alternating bilinear map $\phi_G:\F^n\times\F^n\to\F^m$ by sending $(u, v)\in \F^n\times \F^n$ to $(u^\tr A_1v, \dots, u^\tr A_mv)^\tr$. We shall refer to $\phi_G$ as the \emph{graphical alternating bilinear map} corresponding to $G$.


\paragraph{Independent sets and totally-isotropic spaces.} Let $U$ and $W$ be finite-dimensional vector spaces over a field $\F$. Let $\phi:U\times U\to W$ be an alternating bilinear map. A subspace $V\leq U$ is a \emph{totally-isotropic space} (TI space for short) of $\phi$ if the restriction of $\phi$ to $V \times V$ is the zero map, that is, for any $v, v'\in V$, $\phi(v, v')=0$. 

Some recent works study totally-isotropic spaces as a linear algebraic analogue of independent sets. 
Let $\alpha(\phi)$ denote the maximum dimension among totally-isotropic spaces of $\phi$. In \cite{BCGQS21}, it was shown that the maximum independent set size of a graph $G$ is equal to $\alpha(\phi_G)$; recall that $\phi_G$ is the graphical bilinear map corresponding to $G$. Such a correspondence is generalized to graph independence polynomials \cite{LM05}. In \cite{Qia24}, the totally-isotropic space polynomial was proposed and shown to form a $q$-analogue of graph independence polynomials. 

Based on totally-isotropic spaces, in \cite{Qia23}, Tur\'an and Ramsey problems for multilinear maps were studied, with origins in group theory \cite{BGH87} and geometry \cite{feldman1992linear}. Some open questions from \cite{Qia23} were recently resolved in \cite{CY25,CXY26}.


\paragraph{A random model of alternating bilinear maps over finite fields.}
To prepare for a planted totally-isotropic space problem, we need a random model for alternating bilinear maps. This is realized by noting that after fixing bases of the relevant vector spaces, an alternating bilinear map is represented by a tuple of alternating matrices. Let $\Lambda(n, \F_q)$ denote the linear space of $n\times n$ alternating\footnote{An $n\times n$ matrix $A$ over a field $\F$ is alternating if for any $u\in \F^n$, $u^\tr Au=0$. Over fields of characteristic not $2$, a matrix is alternating if and only if it is skew-symmetric, i.e., $A=-A^\tr$.} matrices over $\F_q$. From a tuple of alternating matrices $\tA=(A_1, \dots, A_m)$, where $A_i\in \Lambda(n, \F_q)$, we can construct an alternating bilinear map $\phi_\tA:\F_q^n\times\F_q^n\to\F_q^m$, defined for $u, v\in \F_q^n$ by
\begin{equation}\label{eq:bilinear-map}
	\phi_\tA(u, v)=(u^\tr A_1v, \dots, u^\tr A_mv)^\tr.
\end{equation}






We then have a natural random model of alternating bilinear maps over finite fields, which was studied, for example, in \cite{LQ17}. 


\begin{definition}\label{def:LinER}
    The \emph{linear Erd\H{o}s--R\'enyi model} $\LinER(n, m, q)$ is a distribution over alternating bilinear maps $\phi_{\tA}:\F_q^n\times\F_q^n\to\F_q^m$, where $\tA\in \Lambda(n, \F_q)^m$ is sampled as follows. 
    \begin{enumerate}
    \item A random alternating matrix tuple $\tA=(A_1, \dots, A_m)\in \Lambda(n, \F_q)^m$ is obtained by sampling each $A_i$ independently, where each $A_i$ is sampled as follows. 
    \item Sample a random alternating matrix $A=(a_{i,j})\in\Lambda(n, \F_q)$ by setting $a_{i,i}=0$ for $i\in[n]$, independently sampling $a_{i,j}$ uniformly at random from $\F_q$ for $1\leq i<j\leq n$, and setting $a_{i,j}=-a_{j,i}$ for $1\leq j<i\leq n$. 
    \end{enumerate}
\end{definition}

The random model $\LinER(n, m, q)$ was recently examined more closely in an ongoing work of Pham, Qiao, Wigderson, and Wigderson \cite{PQWW}. For a random $\phi:\F_q^n\times\F_q^n\to\F_q^m$ sampled from $\LinER(n, m, q)$, a natural question is to estimate $\alpha(\phi)$, the maximum dimension among TI spaces of $\phi$. This question is analogous to estimating $\alpha(G)$ for a random $G$ sampled from $\ER(n, m)$ or $\ER(n, p)$. Estimating $\alpha(\phi)$ with $\phi$ sampled from $\LinER(n, m, q)$ can be traced back to work in group theory in the 1980s \cite{Ols78,BGH87}, which gave a lower bound of $\alpha(\phi)$ using the first-moment method. With a second-moment method \cite{PQWW}, the following was proved.

\begin{theorem}[{\cite{Ols78,BGH87,PQWW}}]\label{thm:TI-dim}
Let $\phi:\F_q^n\times\F_q^n\to \F_q^m$ be a random alternating bilinear map with $m\geq 2$. Then $\alpha(\phi)$, the maximum totally-isotropic space dimension of $\phi$, is $\lfloor(m+2n)/(m+2)\rfloor$ with high probability.
\end{theorem}
\begin{remark}\label{rem:choice-of-m}
Noting that $\lfloor(m+2n)/(m+2)\rfloor\leq 2n/m+1$, we shall assume that we plant in a TI space of dimension $\gg 2n/m$ in this article. An interesting case is when $m=\lceil n/\log n\rceil$. In this case, the maximum dimension among TI spaces is $(2+o(1))\log n$, matching the maximum clique size of the Erd\H{o}s--R\'enyi model $\ER(n, p=1/2)$. 
Therefore, in this paper, we shall mostly work with $\LinER(n, m=\lceil n/\log n\rceil, q)$.
\end{remark}


\paragraph{The planted totally-isotropic space problem.} Based on the above works, it is natural to formulate a planted totally-isotropic space problem in analogy with the planted independent set problem. Note that we shall move from graphical alternating bilinear maps to general alternating bilinear maps, and we will be working over finite fields.

We now introduce a random model of alternating bilinear maps with a random dimension-$d$ totally-isotropic space. 

\begin{definition}\label{def:PlLinER}
    The \emph{$d$-totally-isotropic ($d$-TI) planted linear Erd\H{o}s--R\'enyi model} $\PlLinER(n, m, q, d)$ is a distribution over alternating bilinear maps $\phi_{\tilde\tA}:\F_q^n\times\F_q^n\to\F_q^m$ with a $d$-dimensional totally-isotropic space, where $\tilde\tA\in\Lambda(n, \F_q)^m$ is obtained as follows. 
    \begin{enumerate}
        \item Sample a random alternating matrix tuple $\tA=(A_1, \dots, A_m)\in \Lambda(n, \F_q)^m$ according to the sampling procedure of $\LinER(n, m, q)$ as in \cref{def:LinER}. 
        \item For $i\in[m]$, set the leading principal submatrix of order \(d\) (i.e., the top-left $d\times d$ submatrix) of $A_i$ to $0$, and denote by $\tilde A_i$ the resulting matrix.
        \item Sample \(L\) uniformly at random from \(\GL(n,\F_q)\), and let $\tilde\tA=(L^\tr \tilde A_1L, \dots, L^\tr \tilde A_mL)$. 
    \end{enumerate}
\end{definition}


Note that Step 2 in Definition~\ref{def:PlLinER} ensures that the linear space $\linspan\{e_1, \dots, e_d\}$ is a totally-isotropic space, where $e_i$ denotes the $i$th standard basis vector of $\F_q^n$. Then by Step 3, the alternating bilinear map $\phi_{\tilde\tA}$ admits a TI space spanned by $\{L^{-1}e_1, \dots, L^{-1}e_d\}$. 

Equivalently, one may first sample a uniformly random \(d\)-dimensional subspace \(V\leq \F_q^n\), and then sample a random alternating bilinear map from \(\LinER(n,m,q)\) conditioned on \(V\) being totally isotropic.



The planted totally-isotropic space problem can be formulated as follows. The 
decision version asks to distinguish whether a given alternating bilinear map is sampled from the distribution $\LinER(n, m, q)$ or $\PlLinER(n, m, q, d)$. The search version asks to compute the planted $d$-dimensional TI space from an alternating bilinear map sampled from $\PlLinER(n, m, q, d)$. As discussed in Remark~\ref{rem:choice-of-m}, we will mostly focus on $m=\lceil n/\log n\rceil$, while the choices of $d$ and the precise formulations of the hardness conjectures will be specified after the algorithmic results in Section~\ref{subsec:algorithm} (Conjectures~\ref{conj:search} and~\ref{conj:decision}).

Recall from \cite{BCGQS21} that $\alpha(G)=\alpha(\phi_G)$ for a graph $G$ and its corresponding graphical bilinear map $\phi_G$, establishing a first connection between independent sets and totally-isotropic spaces. However, that result does not imply a correspondence between the planted independent set problem and the planted totally-isotropic space problem, as the distributions of graphs under $\ER(n, m)$ (or $\ER(n, p)$) and of alternating bilinear maps under $\LinER(n, m, q)$ are quite different. Indeed, most alternating bilinear maps are not graphical. Thus, our planted TI problem should be viewed as a linear-algebraic analogue of planted clique, rather than as a direct algebraic encoding of the planted clique distribution. Still, as the reader will see, the research questions and methodologies for the planted totally-isotropic space problem are strongly influenced by those for the planted clique problem. 

\subsection{Algorithms for the planted totally-isotropic space problem}\label{subsec:algorithm}




\paragraph{Review of algorithms for planted clique.} We first recall some key algorithmic results for the planted clique problem in $\ER(n, p=1/2, d)$, where $d\gg (2+o(1))\log n$ is the planted clique size. 

First, for any planted clique size $d$, the planted clique problem can be solved in time $n^{O(\log n)}$. This is done by enumerating all vertex subsets of size $3\log n$ and checking their common neighbors.

Second, when $d\geq c \sqrt{n\log n}$ for some constant $c$, the planted clique can be recovered by listing those vertices with the highest degrees as observed by Ku\v{c}era \cite{Ku95}.

Third, when the planted clique size $d\geq c\sqrt{n}$ for some constant $c$, the celebrated algorithm of Alon, Krivelevich and Sudakov can recover the planted clique in polynomial time using spectral techniques \cite{AKS98}. Since then, algorithms that achieve essentially the same result as \cite{AKS98} have been developed, including the one in \cite{FK03} that utilizes Lov\'asz theta function, and combinatorial algorithms in \cite{FR10,DGP14}. 

The planted clique problem admits a decision version, which is useful for cryptographic purposes. Several works \cite{alon2007testing,hirahara2023hardness} discussed the relation between search and decision versions, leading to the work of Hirahara and Shimizu \cite{HS24} who established an optimal equivalence between these search and decision variants.

\paragraph{Algorithms for planted TI space.} We now turn to the planted totally-isotropic space problem for $\PlLinER(n, m, q, d)$ where $m=\lceil n/\log n\rceil$ and $d\gg (2+o(1))\log n$. 

As a baseline, note that a natural brute-force algorithm would enumerate every $d$-dimensional subspace of $\F_q^n$ and check if it is totally-isotropic for the input alternating bilinear maps. This algorithm runs in time $q^{O(dn)}$. 

As in the planted clique problem, we can improve the brute-force algorithm by noting that we only need to examine subspaces of dimension, say, $\lceil 3\log n\rceil$, leading to an algorithm in time $q^{O(n\log n)}$ (Proposition~\ref{prop:general-algorithm}). Once such a small planted subspace is identified, the full planted TI space can be recovered efficiently. Note that this still gives an exponential-time algorithm, in contrast to the planted clique problem where the general case is assumed to be quasipolynomial-time hard.

We then examine how large the planted TI space dimension $d$ could be to allow for a polynomial-time recovery algorithm or distinguishing algorithm. Unlike the planted clique problem, even getting a better-than-brute-force algorithm for $d\geq\lceil n/C\rceil$ for some constant $C\geq 3$ seems tricky. 

Our current results suggest that
$d\approx n/2-o(n)$ may mark a threshold for polynomial-time recovery. First, inspired by the degree-based approaches for the planted clique problem, we analyze its counterpart in tensors—the rank of matrix slices. This yields a straightforward polynomial-time distinguisher for $d > n/2$. Moreover, for $d=n/2-c$ with $c=O(\sqrt{n/\log n})$, this also provides a distinguishing algorithm whose running time is polynomial in $n$ and exponential in $c$ (see \Cref{sec:distinguisher}). We then obtain a polynomial-time algorithm that recovers the planted TI space for $d\geq n/2$ (Proposition~\ref{prop:n-over-2}), by resorting to the solutions of the non-commutative rank problem \cite{GGOW20,IQS18,HH21}.

Finally, we note that the planted TI space problem can be formulated as solving a system of polynomial equations. In Section~\ref{subsec:grobner}, we report our experiments based on polynomial equation solving in Magma \cite{Magma}.

All the above evidence indicates that the planted TI space problem seems exponential-time hard, in contrast to the quasipolynomial-time hardness of the planted clique problem. Formally, we propose the following conjectures. 

\begin{conjecture}[The hardness assumption for finding planted TI spaces]\label{conj:search}
    There exists a constant $C\geq 3$ such that no randomized algorithm running in time $q^{o(n)}$ can recover the planted TI space for a random $\phi$ sampled from $\PlLinER(n, m=\lceil n/\log n\rceil, q, d)$ for $d=\lceil n/C\rceil$.
\end{conjecture}

\begin{conjecture}[The hardness assumption for distinguishing random alternating bilinear maps with or without TI spaces]\label{conj:decision}
    There exist constants $C\geq 3$ and $\varepsilon >0$ such that no randomized algorithm running in time $q^{o(n)}$ can distinguish whether $\phi:\F_q^n\times\F_q^n\to\F_q^m$ is sampled from $\LinER(n, m=\lceil n/\log n\rceil, q)$ or from $\PlLinER(n, m=\lceil n/\log n\rceil, q, d=\lceil n/C\rceil)$ with advantage at least $\varepsilon$. Formally, for any randomized algorithm $D$ in time $q^{o(n)}$, we have 
    $$
    \Big|\Pr_{D, \phi\sim \LinER(n, \lceil n/\log n\rceil, q)}[D(\phi)=1]-\Pr_{D, \phi\sim \PlLinER(n, \lceil n/\log n\rceil, q, \lceil n/C\rceil)}[D(\phi)=1]\Big|<\varepsilon.
    $$
\end{conjecture}



\subsection{Cryptographic applications of the planted tensor problem}\label{subsec:crypto}

Our algorithmic investigations indicate that the planted totally-isotropic space problem for $\LinER(n$, $m=\lceil n/\log n\rceil$, $q$, $d=\lceil n/C\rceil)$ for some constant $C\geq 3$ could be exponentially hard. It is then of natural interest to see if any cryptographic schemes could be devised based on the planted TI space problem or the more general planted random tensor problem. We demonstrate that this is possible for the private simultaneous messages and secret-sharing schemes equipped with a public information model presented in a recent work by Abram, Beimel, Ishai, Kushilevitz, and Narayanan \cite{ABIKN23}. In this subsection, we first review the two cryptographic primitives and their constructions under the planted graph assumptions. Then we formalize our planted tensor assumptions and explain how to use them for our constructions in detail. 

\paragraph{Private simultaneous messages (PSM).} A \textit{private simultaneous messages} (PSM) protocol, first introduced in \cite{FKN94} and named by Ishai
and Kushilevitz \cite{ishai1997private}, is a cryptographic primitive that enables $d$ parties to jointly and non-interactively compute a function $f$ on their respective private inputs, $x_1, x_2, \dots, x_d$. The protocol ensures that a referee, upon receiving the transmitted encodings from all parties, can correctly evaluate the function's output $f(x_1, x_2, \dots, x_d)$, while learning nothing else about the original inputs. PSM can be viewed as a particularly restricted form of non-interactive secure computation, in which each input holder sends only one simultaneous message to a referee; this connection to broader non-interactive secure multiparty computation has been explored in later works such as \cite{beimel2014non}.


\paragraph{Secret-sharing schemes for forbidden graph access
structures (FGSS).} A \textit{secret-sharing scheme} \cite{Sha79,Bla79,ISN89} is a cryptographic primitive for distributing a secret among multiple parties in such a way that only certain groups of them can recover it. Specifically, the scheme's access structure specifies which sets of parties are authorized to fully reconstruct the secret and which sets are forbidden from learning anything about it. Other sets, not explicitly defined in the access structure, fall into an indeterminate zone where they may obtain nothing, partial information, or even the entire secret. A notable variant, introduced in \cite{SS97}, is the \textit{forbidden graph secret-sharing} (FGSS) scheme, where the access structure is represented by a graph. In this graph, each vertex corresponds to a party, and any two parties can recover the secret if and only if their corresponding vertices are adjacent, while non-adjacent pairs are guaranteed to learn nothing from their secret shares. Note that for larger subsets of more than two parties, the scheme provides no guarantee as to whether the secret can be recovered.

\paragraph{From planting in TI spaces to planting in general alternating bilinear maps.} To enable cryptographic applications, we need to accommodate planting in general bilinear maps, not just the zero one as in the totally-isotropic space case. 
\begin{definition}\label{def:PlRandomLinER}
    Let $\tB=(B_1, B_2, \dots, B_m)\in \Lambda(d, \F_q)^m$. 
    The \emph{$\tB$-planted linear Erd\H{o}s--R\'enyi model} $\PlLinER(n, m, q, \tB)$ is a distribution over alternating bilinear maps $\phi_{\tilde\tA}:\F_q^n\times\F_q^n\to\F_q^m$, where $\tilde\tA\in\Lambda(n, \F_q)^m$ is obtained as follows. 
    \begin{enumerate}
        \item Sample a random alternating matrix tuple $\tA=(A_1, \dots, A_m)\in \Lambda(n, \F_q)^m$ from $\LinER(n, m, q)$ as in \cref{def:LinER}. 
        \item (Embedding) For $i\in[m]$, set the $d$th order leading principal submatrix (i.e., the top-left $d\times d$ submatrix) of $A_i$ to be $B_i$, and denote by $\tilde A_i$ the resulting matrix.
        \item (Hiding) Sample a random $L\in\GL(n, \F_q)$, and let $\tilde\tA=(L^\tr \tilde A_1L, \dots, L^\tr \tilde A_mL)$. 
    \end{enumerate}

    We define a distribution $\PPlLinER(n, m, q, \tB)$ over pairs of an alternating bilinear map and an invertible matrix that follows the above steps and outputs $(\tilde \tA, L)$ where $\tilde \tA\in \Lambda(n, \F_q)^m$ and $L\in \GL(n, \F_q)$ as above. 
\end{definition}

We can define the $\tB$-planted random tensor problem as recovering the planted embedding of $\tB$ from $\tilde\tA$ sampled from $\PlLinER(n, m, q, \tB)$, as well as its decision version similarly to those for the planted TI space problem. By the common belief (e.g., in \cite{ABIKN23}) that the planted clique problem represents the most detectable instance among other types of planted graphs, we choose the same parameters of the planted random tensor conjecture as those in the planted TI space conjectures (\cref{conj:search} and \cref{conj:decision}). 

Note that we recover Definition~\ref{def:PlLinER} by setting $\tB$ to be the tuple of all-zero $d\times d$ matrices in Definition~\ref{def:PlRandomLinER}.

\paragraph{Our results: PSM and FGSS based on planted tensors.} Before stating our results, we briefly explain the underlying conjectures we need. Our constructions rely on two formal assumptions regarding the exponential-time hardness of detecting planted structures in random tensors. 

The first is the \textit{Planted Alternating Tensor with 2-Hints} (Conjecture~\ref{conj:\PRSwH}), which states that no sub-exponential time adversary can distinguish a random alternating tensor from one containing a fixed planted subtensor, even when provided with two hint vectors related to the hidden basis transformation. The second is the \textit{Planted Random Alternating Tensor with 2-Hints}  (Conjecture \ref{conj:\PRATwTwoH}), a weaker assumption where the planted subtensor is sampled randomly rather than being fixed. These assumptions guarantee exponential hardness of our subsequent PSM and FGSS protocols; see \cref{remark:justification} for the justification in this hint case.

\begin{conjecture}[Informal - Planted Alternating Tensor with 2-Hints]\label{1}
There exists some constant $C\geq3$ such that the following holds. For a large enough $n\in \N$, set $m=\lceil n/\log n\rceil$, and $d=\lceil n/C\rceil$. Fix $\tB\in\AT(d,m,\F_q)$. Then no randomized algorithm in time $q^{o(n)}$ can distinguish whether a random $\phi:\F_q^n\times\F_q^n\to\F_q^m$ is sampled from $\LinER(n, m=\lceil n/\log n\rceil, q)$ or $\LinER(n, m=\lceil n/\log n\rceil, q, \tB)$ with non-negligible probability, even with two vectors in the row space of the inverse hidden transformation as hints.
\end{conjecture}

\begin{theorem}
    Suppose $d=2r$ is an even number, and assume Conjecture~\ref{1} holds for parameters $C, n, m, d$ there. Then there exists a $2$-party PSM construction with public information for bilinear functions $f:\F_q^r \times \F_q^r \to \F_q^m$ achieving $(T,\varepsilon)$-security for every adversary running in time $T=q^{o(n)}$, with per-party message size $C \cdot d \log q$ and public information size  $O\left(\frac{d^3\log q}{\log d}\right)$.
\end{theorem}

We now present a slight variant of Conjecture~\ref{1} by allowing the fixed $\tB$ there to be random.
\begin{conjecture}[Informal - Planted Random Alternating Tensor with 2-Hints]\label{2}
There exists some constant $C\geq3$ such that the following holds. For a large enough $n\in \N$, set $m=\lceil n/\log n\rceil$, and $d=\lceil n/C\rceil$. Let $\tB\smpl\mathcal{AT}(d,m,\F_q)$. Then no randomized algorithm in time $q^{o(n)}$ can distinguish whether a random $\phi:\F_q^n\times\F_q^n\to\F_q^m$ is sampled from $\LinER(n, m=\lceil n/\log n\rceil, q)$ or $\LinER(n, m=\lceil n/\log n\rceil, q, \tB)$ with non-negligible probability, even with two vectors in the row space of the inverse hidden transformation as hints.
\end{conjecture}

\begin{theorem}
    Assume Conjecture \ref{2} holds for parameters $C, n, m, d$ there. Then there exists a secret-sharing scheme for forbidden graph access
structures for sharing $(d\log q/\log d)$-bit secrets over $d$ parties that achieves $(T, 2\varepsilon)$-security, with share size $C\cdot d\log q$, and public information size $O\left(\frac{d^3\log q}{\log d}\right)$, for every adversary running in time $T=q^{o(n)}$.
\end{theorem}

\begin{remark}[Justification for the security with two hints]\label{remark:justification}
In planted cliques, it is believed that revealing two vertices in the planted clique would not speed up the recovery of the clique. This is because a pair of vertices in the clique looks indistinguishable from a pair of vertices outside the clique. No known algorithms for planted clique could exploit knowing constantly many vertices in the clique. This supports the use of planted cliques with hints in \cite{ABIKN23}.

Similarly, for the planted TI space problem, revealing two vectors in the planted TI space seems not helpful to recover the TI space. More specifically, let $\phi:\F_q^n\times\F_q^n\to\F_q^m$ be sampled from $\LinER(n, m=\lceil n/\log n\rceil, q, d)$ with $d=\lceil n/C\rceil$ for some constant $C\geq 3$. Let $V\leq\F_q^n$ be the planted TI space. Let $v, v'\in \F_q^n$ (not necessarily in $V$). Construct $\phi_v, \phi_{v'}:\F_q^n\to\F_q^m$. As $m=\lceil n/\log n\rceil \ll n$, the left kernels of $\phi_v$ and $\phi_{v'}$ intersect in a subspace of dimension $n-2\cdot \frac{n}{\log n}$, which is large enough to accommodate a subspace of dimension $\lceil n/C\rceil$. In particular, whether $v$ and $v'$ are in $V$ or not cannot be reflected by examining such ``neighborhood'' information, similar to the graph setting.

If some algorithm could successfully exploit knowing two vectors in the planted TI space and run in time $q^{O(n)}$, then we would obtain an algorithm in time $q^{O(n)}$, as enumerating two vectors has cost $q^{2n}$. However, we don't know how to achieve this, and the best algorithm we know for the general case still runs in time $q^{O(n\log n)}$ (\cref{prop:general-algorithm}).
\end{remark}

Compared with the construction of~\cite{ABIKN23} in the same setting, our tensor-based protocols enjoy two advantages. First, both our PSM and FGSS constructions achieve security against subexponential-time adversaries under our conjectures, more precisely against adversaries running in time \(q^{o(n)}\), whereas the construction of~\cite{ABIKN23} only guarantees security against quasi-polynomial-time adversaries. As a result, to attain the same level of security, their parameters must scale superpolynomially relative to ours. In particular, for the FGSS construction, our scheme maintains compact private shares in the multi-bit secret-sharing setting. For the PSM construction, the difference is more subtle: the planted-subgraph protocol of~\cite{ABIKN23} is tailored to a single Boolean output bit, whereas our construction natively supports bilinear maps with multi-coordinate output in one planted tensor instance. Thus, while the communication is comparable under a natural bit-wise repetition, our PSM protocol handles multi-bit output in a more direct algebraic way under a conjecturally stronger security foundation. We will quantify this improvement more precisely in the next subsection after presenting our concrete constructions.



\subsection{Overview of cryptographic protocols based on planted tensors}

In this section we give an overview of our protocols for PSM and FGSS. For this, we shall review the corresponding protocols by Abram, Beimel, Ishai, Kushilevitz, and Narayanan \cite{ABIKN23}.

\subsubsection{Review of PSM and FGSS with public information from planted graphs} To optimize communication cost, \cite{ABIKN23} introduced a variant of the standard model for PSM and FGSS with \textit{public information}. In this model, an offline setup phase generates a large, publicly accessible data object, alongside compact private secrets for each party. This public object is structured to act as a complex computational landscape. While it reveals nothing on its own under appropriate computational assumptions, parties can use their secrets as precise coordinates to navigate this landscape and extract a specific result. By contrast, any adversary lacking these coordinates would remain lost. This division of resources from the public information model makes the ``planted problem'' paradigm a very natural fit. 

In \cite{ABIKN23}, the public data object is instantiated as a  randomly generated ambient graph with a random subgraph (not necessarily a clique) planted inside, and the compact private secrets act as individual injection mappings. Note that if the planted subgraph is fixed to be a clique, the underlying problem will degenerate to the planted clique problem. The authors of \cite{ABIKN23} then construct a $2$-party PSM protocol for Boolean functions $f:[d]\times [d]\to\{0,1\}$. First, $f$ is encoded as a bipartite graph $H$ with $d$ vertices on each side, where an edge $(x,y)$ exists if and only if $f(x,y)=1$. The setup samples a large random graph $G$ and plants $H$ inside it by following the clique planting procedure. 
Then $G$ is published as public information while giving both parties the embedding of $H$ into $G$ as secrets. On inputs $x$ and $y$, each party sends the location of its corresponding vertex in $G$, and the referee outputs $1$ if the two revealed vertices are adjacent in $G$ and outputs $0$ otherwise. 

The construction of FGSS in \cite{ABIKN23} is more involved but follows a similar graph-structural approach, i.e., making a planted subgraph whose $d$ vertices correspond to $d$ parties for the secret-sharing. To share a single-bit secret from $\{0,1\}$, the dealer first samples a random small graph $H$ from $\mathcal{G}(d,p=1/2)$ and plants it into a larger random graph. Let the resulting graph after planting be $G$. Each party $i$ receives the location $\phi(i)$ of its corresponding vertex in the embedding of $G$ as a respective private share. Then, based on the forbidden graph $Q$ defining the access structure, the dealer derives a new graph $H'$ by removing from $H$ all edges that are not present in $Q$. Intuitively, this ``erases'' the information carried by the edge between pairs of parties (i.e., non-adjacent vertices in $Q$) that are forbidden to recover the secret. The public information is then released as $(H', G')$, where $G'$ is the same as $G$ if the secret bit is $1$, and $G'$ is the complement graph of $G$ if the secret bit is $0$. 

The secret recovery for parties $i,j$ is to compare the edge $\{i,j\}$ in $H'$ with the edge $\{\phi(i),\phi(j)\}$ in $G'$: if they are both present or both absent in $H'$ and $G'$, the output is $1$, otherwise it is $0$. For authorized pairs, the edge or non-edge is retained in $H'$ and always matches its embedding in $G$, ensuring correct recovery. However, for forbidden pairs, $\{i,j\}$ is always a non-edge in $H'$, while the existence of edge $\{\phi(i),\phi(j)\}$ in $G'$ is fully random with probability $1/2$, because it depends on the existence of edge $\{i,j\}$ in the planted graph $H$ that is sampled from $\mathcal{G}(d,p=1/2)$.

    The security of the above PSM and FGSS constructions relies on the computational intractability of detecting the embedded small graphs. In particular, for the FGSS scheme with public information, \cite{ABIKN23} introduced \textit{the Planted Random Subgraph with Hints} (\PRSwH) assumption. The ``hints'' parameter explicitly reflects the protocol mechanics where two vertex locations are leaked to the forbidden pairs during the secret recovery phase. The asymptotic efficiency of these protocols is directly dictated by $n$, the total number of vertices in the public ambient graph $G$. Because each party's message or private share is simply the index of a specific vertex in $G$, the communication cost is exactly $\log n$ bits. They further conjecture that under the \PRSwH assumption, the planted random subgraph of order $d$ remains computationally hidden in $G$ when $n = d^{1+\delta}$ for any arbitrarily small constant $\delta > 0$, against all non-uniform adversaries running in $d^{o(\log d)}$ time. Consequently, for $d$ parties sharing a single-bit secret, the required private share size is tightly compressed to $\log(d^{1+\delta}) = (1+\delta)\log d$. This beats the best-known, and typically the best possible, information-theoretic FGSS scheme with share size $2^{\tilde{O}(\sqrt{\log d})}$ \cite{LVW17}. It also offered an improvement over the best-known computational FGSS scheme (without public information) with share size $\operatorname{poly}(\log d)$ that assumes the existence of one-way functions with sub-exponential security \cite{ABI+23}.

Nevertheless, the work of \cite{ABIKN23} was intrinsically tied to the structure of graphs. Graphs are fundamentally representations of pairwise relationships over the binary field $\F_2$, which naturally limits the scope of their constructions to Boolean functions and yields security only against quasipolynomial-time adversaries. This motivates us to extend this ``planted'' paradigm beyond graphs to encompass scenarios with more secure algebraic structures for at least the same cryptographic applications. As discussed in the previous subsections, the planted TI space problem provides a promising candidate in this direction. Analogous to how the planted clique problem extends to planted random subgraph problems, we now generalize the planted TI space problem to the planted random subtensor setting and then use it to construct our PSM and FGSS with public information for exponential security at a comparable cost.

\subsubsection{Our construction of PSM and FGSS from planted tensors}

Similar to prior cryptographic constructions based on planted graphs, the mechanics of our protocols also leak partial information to the adversary. In the FGSS scheme from planted graphs, when unauthorized parties attempt to recover the secret, they have their specific vertex locations within the large ambient graph. This structural leakage is encoded in their planted random subgraph with hints (\PRSwH) assumption, where the hints are precisely these exposed vertex indices. In our FGSS scheme, the private shares distributed to the parties are algebraically tied to the random change-of-basis matrix $L$ used during the hiding step of planting. When unauthorized parties attempt to evaluate the tensor, their shares leak partial information about this underlying hiding transformation. Specifically, these leaked hints manifest as two vectors in the row space of the inverse transformation $L^{-1}$. This will be seen more explicitly below when we detail the concrete construction of our FGSS scheme. We will formalize this scenario as the \textit{Planted Random Alternating Tensor with 2-Hints} (\PRATwTwoH) assumption in \cref{sec:assumptions}.
We now present our constructions of PSM and FGSS based on planted tensors.

\paragraph{PSM with public information from planted tensors.} 
To be consistent with \cite{ABI+23}, we first explain how to plant in a general bilinear map $f:\F_q^r\times \F_q^r\to \F_q^m$ (not necessarily alternating). 
A bilinear function $f:\F_q^r\times \F_q^r\to \F_q^m$ can be represented by a tuple of matrices $H=(H_1,\dots,H_m)$, where the $k$th output coordinate is given by $x^\tr H_k y$. To accommodate the alternating bilinear map setting, we first convert this tuple into an alternating matrix tuple $\tB=(B_1,\dots,B_m)$, where $B_i=\begin{bmatrix}
   0 & H_i \\
   -H_i^\tr & 0
\end{bmatrix}$. 
We can then use $\tB$ in Definition~\ref{def:PlRandomLinER} as a way to hide $f$ in a random alternating bilinear map $\tA$.
The resulting tensor $\tA$ is published as the public information. 

The two parties receive short private values derived from the hidden basis transformation. These private values allow them to convert their inputs $x,y\in\F_q^r$ into two vectors $\hat x, \hat y\in\F_q^n$ that are aligned with the hidden planted block inside $\tA$. The referee then evaluates the public tensor on these two vectors and outputs $\hat x^\tr A_1 \hat y_1,\dots,\hat x_0^\tr A_m \hat y$. By construction, the result is exactly $f(x,y)$. Intuitively, the public tensor $\tA$ plays the role of a large obfuscated algebraic landscape, while the two private values serve as compact coordinates that let the two parties navigate this landscape and reach the planted tensor $\tB$. Security is based on the hardness of distinguishing a random alternating tensor from one containing a hidden planted subtensor, even when the adversary also sees the two transmitted vectors. 

In this way, we obtain a $2$-party PSM protocol with public information for bilinear functions over finite fields, extending the planted graph paradigm of \cite{ABIKN23} to a more generalized algebraic setting. This is because a bilinear function $f:\F_q^r\times \F_q^r\to \F_q^m$ may encode a Boolean function $b:[r]\times [r]\to \{0, 1\}$ by the following rule: if $b(i, j)=1$, then $f(e_i, e_j)=u$ for some non-zero $u\in \F_q^m$; if $b(i, j)=0$, then $f(e_i, e_j)=\bzero$ where $\bzero$ denotes the zero vector. 

\paragraph{FGSS with public information from planted tensors.}
To better illustrate our construction of the secret-sharing scheme for $d$ parties whose access structure is defined by a forbidden graph $Q_d=([d], E)$, we first consider a relatively simple version step by step: 
\begin{enumerate}[leftmargin=*,label=Step \arabic*., start=0]
    \item Index each of the $d$ parties from $1$ to $d$, and let each party know its own index.
    \item The dealer uniformly samples a random small tensor $\tB$ based on $\mathcal{AT}(d, m, q)$, and saves a copy $\widehat{\tB}=\tB$ for the upcoming changes. Denote by $(\widehat{B}_1,\dots,\widehat{B}_m)\in\Lambda(d,\F_q)^m$ the frontal slices of $\widehat{\tB}$. 
    \item According to the forbidden graph $Q_d$, the dealer zeros out the value of $\widehat{B}_k(i,j)$ for all $(i,j)$ satisfying $i\nsim j$ in $Q_d$ and all $k\in[m]$. (Intuitively, this step is doing ``puncturing'' at each position $(i,j)$ from the frontal direction through tensor $\widehat{\tB}$ for all the non-edges $\{i,j\}$ in $Q_d$. This is also analogous to the FGSS construction based on planted graphs, where the dealer removes from the randomly sampled $H$ all edges that are absent in $Q_d$.)
    \item The dealer uniformly samples a large tensor $\tA$ based on $\mathcal{AT}(n, m, q)$ and plants the initially sampled $\tB$ into $\tA$ to produce the resulting tensor denoted by $\tilde{\tA}\in\AT(n, m,\F_q)$. (Note that planting is not simply a procedure of ``embedding'' but also includes ``hiding'' with some randomly sampled transformation $L\in\GL(n,\F_q)$. In other words, the top-left $(d\times d\times m)$-block of $L^{-\tr}\tilde{\tA}L^{-1}$ exactly recovers the planted $\tB$.)
    \item The dealer sets the public information to $(x\cdot\widehat{\tB},\tilde{\tA})$ for a $(q-1)$-ary secret $x\in\{1,\dots,q-1\}$ hidden as a scalar, and assigns party \(i\) the share \(s_i:=L^{-1}[I_d \ 0]^\tr e_i\). (This is a more general way to encode a larger secret compared to the planted graph-based FGSS, where the public information is restricted to be either the large graph or its complement for a single-bit secret.)
\end{enumerate}
Denote $J_d =[I_d \ 0]^\tr \in \F_q^{n\times d}$. It can be seen that if each party $i$ is given $s_i=L^{-1}J_de_i$, then any two parties $i$ and $j$ with $i\sim j$ in $Q_d$ can recover the secret by computing $(L^{-1}J_de_i)^{\tr}\tilde{\tA}(L^{-1}J_de_j)=e_i^{\tr}\tB e_j=e_i^{\tr}\widehat{\tB} e_j$, so they can easily compare $e_i^{\tr}\widehat{\tB} e_j$ with the public $e_i^{\tr}(x\cdot\widehat{\tB}) e_j$ to recover the secret $x$, while any two parties $i^\prime$ and $j^\prime$ with $i^\prime\nsim j^\prime$ in $Q_d$ can only get a zero vector $e_{i^\prime}^{\tr}(x\cdot\widehat{\tB}) e_{j^\prime}$ and thereby gaining no information about $x$ unless they can recover $\tB$ with their limited hints of $L^{-1}J_de_{i'}$ and $L^{-1}J_de_{j'}$.

There are two aspects of improvements we can make to the secret-sharing scheme proposed above. One aspect is that we do not want adversaries to get a vector that directly reveals certain information about $L^{-1}$ (e.g., $L^{-1}J_de_i$) but something that looks more random (e.g., $L^{-1}J_du_i$ where $u_i\smpl\F_q^d\setminus \{0\}$). To this end, the dealer can sample an additional transformation $L^{\prime}\in\GL(d,\F_q)$ and define $L^{\prime\prime}=J_dL^\prime\in\F_q^{n\times d}$. This is used for adding obfuscation to the initialization of $\widehat{\tB}$, which is not an identical copy of $\tB$ anymore, but sets to $\widehat{\tB}=(L^{\prime})^\tr \tB L^{\prime}$. With this new $\widehat{\tB}$, the subsequent steps remain the same, except that the dealer needs to allocate the secret share for party $i$ as $L^{-1}L^{\prime\prime}e_i$ instead of $L^{-1}J_de_i$. Similar computations as before show that any two parties $i$ and $j$ with $i\sim j$ in $Q_d$ can get $(L^{-1}L^{\prime\prime}e_i)^{\tr}\tilde{\tA}(L^{-1}L^{\prime\prime}e_j)=(L^{\prime}e_i)^{\tr}\tB(L^{\prime}e_j)=e_i^\tr \widehat{\tB} e_j$ and compare it with the public $e_i^{\tr}(x\cdot\widehat{\tB}) e_j$, while any two parties $i^\prime$ and $j^\prime$ with $i^\prime\nsim j^\prime$ in $Q_d$ get a random vector $e_{i^\prime}^\tr {L^\prime}^\tr \tB L^\prime e_{j^\prime}$ and the public $e_{i^\prime}^{\tr}(x\cdot\widehat{\tB}) e_{j^\prime}=0$. The other aspect of improvement is about the length of secret that can be shared. It is straightforward to boost the secret size from $(q-1)$ to $q^m$ in this way: rather than hiding a scalar $x\in\{1,\dots,q-1\}$ in the entire tensor $x\cdot\tB$, we consider a vector $x:=(x_1,\dots,x_m)^\tr\in\F_q^m$ and additively hide each component by updating $e_i^\tr\tilde{\tB} e_j=x+ e_i^\tr\widehat{\tB} e_j$ for all $i\sim j$ in $Q_d$, so any authorized pair of parties can recover the secret by comparing the components one by one and thereby recovering the additive offset hidden in each component. The final version of our secret-sharing schemes is depicted in \cref{protocol:secret-sharing} in detail.

\paragraph{Comparisons of protocols based on planted subgraphs and planted subtensors.} Recall that planted subgraph problems assume that the planted subgraph order $D$ and the ambient graph order $N$ are related as $D=N^{1/2-\delta}$ for some small constant $\delta$, and the hardness assumption is against $N^{o(\log N)}$-time adversaries. On the other hand, for planted tensor problems, the setting is to plant in $\tB\in \AT(d, m, q)$ to $\tA\in \AT(n, m, q)$ where $d=\lceil n/C\rceil$ for some constant $C\geq 3$ and $m=\lceil n/\log n\rceil$, and the hardness assumption is against $q^{o(n)}$-time adversaries. 

Therefore, to achieve the same security level, we would have $N^{a\cdot \log N}=q^{b\cdot n}$ where $a$ and $b$ are positive constants determined by the best algorithms for planted cliques and planted TI spaces, respectively. This yields that 
\begin{equation}\label{eq:logN}
\log N=\Theta(\sqrt{n\log q}).
\end{equation}

On the other hand, in cryptographic applications, the protocols based on planted subgraphs have the public information size (the graph) as $N^2$, and the communication costs as $\log N$. The protocols based on planted tensors have the public information size (the tensors) as $n^2m\log q=n^3\log q/\log n$, and the communication costs as $n\log q$. By Equation~\ref{eq:logN}, we see that the graph setting public information size is $\binom{N}{2}=2^{\Theta(\sqrt{n\log q})}$, which is (moderately) exponential in the public information size in the tensor setting. On the other hand, the communication costs of protocols based on planted tensors are $n\log q=\Theta((\log N)^2)$, that is a polynomial increase of the communication costs of protocols based on planted subgraphs.

Finally, we note that compared to protocols based on planted subgraphs primarily focusing on single bits, protocols based on planted tensors naturally transmit $m\log q$ bits. Although there are generic transformations based on symmetric encryption that can remove this restriction \cite{ABI+23,ABIKN23}, at comparable message or share sizes, this provides only sub-exponential security under subexponentially hard one-way functions. The main benefit of our new conjecture is that it enables exponential security for the construction.

\subsection{Discussions and open problems}

In this paper we introduced the planted totally-isotropic space problem as a linear algebraic analogue of the planted clique problem. We carried out an initial investigation into this problem and current evidence suggests that this problem could be exponential-time hard. Based on this, we demonstrated some cryptographic protocols based on planted cliques can be adapted to be based on planted totally-isotropic spaces. 

Several open questions can be raised immediately. First, we need to improve our understanding of $\LinER(n, m=\lceil n/\log n\rceil, q, d=\lceil n/C\rceil)$ for $C\geq 3$. Our current best algorithm runs in time $q^{O(n\log n)}$, so even getting an algorithm in time $q^{O(n)}$ would be of great interest. Second, it is an open problem to devise public-key encryption schemes based on planted TI spaces, in line with the works of \cite{ABW10,GHJS25}. Third, our understanding regarding relations between the search and decision versions of the planted TI space problem is quite limited, unlike the planted clique problem \cite{HS24}.

We also note a close relation of this work with an algorithmic problem underlying the Unbalanced Oil and Vinegar (UOV) digital signature scheme \cite{KP99}. Briefly speaking, the following algorithmic problem is related to the key-recovery problem of UOV: plant in a random totally-isotropic space of dimension $d$ in an $m$-tuple of $n\times n$symmetric matrices over $\F_q$, and recover this TI space. In the literature of UOV (see \cite{UOV2025} and references therein), it is usually set that $m=d=n/C$ for some constant $C> 2$. On the other hand, we use alternating matrices instead of symmetric, and we set $m=\lceil n/\log n\rceil$ instead of linear in $n$. Furthermore, our research is largely inspired by a comparison with the planted clique problem as discussed. It will be interesting to compare these two settings more closely. 

\paragraph{Organization.} The remainder of this paper is organized as follows. In \cref{sec:prelim}, we set up our notation in this paper. In \cref{sec:algorithm}, we provide a detailed treatment of the planted totally-isotropic space problems to support our conjectures. In \cref{sec:assumptions}, we generalize the planted totally-isotropic space problems to a family of planted alternating tensor problems for the purpose of cryptographic applications. In \cref{sec:PSM,sec:secretsharing}, we present the main construction and security proof for our PSM protocols and secret sharing schemes, respectively. 

\section{Preliminaries}\label{sec:prelim}


\paragraph{Graphs.} For $n \in \mathbb{N},[n]:=\{1,2, \ldots, n\}$. Let $G=([n],E)$ be a simple and undirected graph on the vertex set $[n]$, and the edge set $E\subseteq \binom{[n]}{2}$. For $i,j\in[n]$, we use $i\sim j$ and $i\nsim j$ to denote $\{i, j\}\in E$ and $\{i, j\}\notin E$, respectively.

\paragraph{Matrices.} Let $\mathbb{F}_q$ be the finite field of order $q$. Let $\mathrm{M}(n \times m, \F_q)$ be the linear space of $n\times m$ matrices over $\F_q$, and $\mathrm{M}(n, \F_q):=\mathrm{M}(n \times n, \F_q)$. Given $A\in \M(n\times m, \F_q)$, the $(i,j)$th entry of $A$ is denoted by $A(i,j)\in \F_q$. 

Let $\GL(n, \F_q)$ be the group of $n \times n$ invertible matrices over $\F_q$. We use $\GL(d\times n, \F_q)$ to denote the variety of rank-$\min(d, n)$ matrices in $\M(d\times n, \F_q)$. 

For a matrix $A \in \GL(n, \F_q)$, let $A^{-1}$ be its inverse. For a matrix $B \in \M(n, \F_q)$, let $B^{\tr}$ be its transpose. We use $\bzero_{d\times n}$ to denote an $d\times n$ all-zero matrix, and $\bzero_{d}$ for a $d\times d$ all-zero matrix.

\paragraph{Vector spaces.} We use $\F_q^n$ to denote the vector space of length-$n$ column vectors over $\F_q$. 
The $i$th standard basis vector of $\F_q^n$ is denoted by $e_i$.

Let $V$ be a vector space. We use $\Gr(d, V)$ to denote the Grassmannian of $d$-dimensional subspaces of $V$. 


\paragraph{Tensors, bilinear maps, and group actions on them.} We use $\mathrm{T}(\ell \times n \times m, \F_q)$ to denote the linear space of $\ell\times n\times m$ tensors
over $\F_q$. We use the fixed-width teletype font for tensors and matrix tuples, like $\tA$, $\tB$, etc. Given $\tA\in \T(\ell\times n\times m, \F_q)$, the $(i,j,k)$th entry of $\tA$ is denoted by $\tA(i,j,k)\in \F_q$. The all-zero tensor in $\mathrm{T}(\ell \times n \times m, \F_q)$ is denoted by $\mathtt{0}_{\ell\times n\times m}$, i.e., $\mathtt{0}_{\ell\times n\times m}(i,j,k)=0$ for all $i \in[\ell], j \in[n], k \in[m]$. 

We can slice $\tA$ along one direction, turning the tensor into several matrices called \textit{slices} of $\tA$. In this paper, we mainly work with slicing along the third coordinate to obtain the \emph{frontal slices}, namely $m$ matrices $A_1, \dots, A_m\in \M(\ell\times n, \F_q)$, where $A_k(i,j)=\tA(i,j,k)$ for each $k\in[m]$. We also use $\tA^\tr$ to denote the tensor whose frontal slices are the transpose of those of $\tA$, i.e., $\tA^\tr(i,j,k) = \tA(j,i,k)$. We sometimes abuse the notation of tensors and matrix tuples, e.g., writing $\tA=(A_1,\dots,A_m)$, as a matrix tuple can define a tensor by specifying its frontal slices.


We consider tensors and bilinear maps as interchangeable notions for the following reason: from any tensor $\tA\in \T(\ell\times n\times m, \F_q)$, by treating each frontal slice of $\tA$ as a bilinear form, we can construct a corresponding bilinear map $\phi_\tA:\F_q^\ell\times \F_q^n\rightarrow \F_q^m$ sending $(u,v)\in \F_q^\ell\times \F_q^n$ to $(u^{\tr}A_1v,\dots,u^{\tr}A_m v)^{\tr}\in\F_q^m$. For convenience, we denote $u^{\tr}\tA v:=\phi_\tA(u,v)$. Conversely, from any bilinear map $\phi:\F_q^\ell\times \F_q^n\rightarrow \F_q^m$, we can construct a corresponding tensor $\tA_{\phi}$ by defining $\tA_{\phi}(i,j,k)$ as the $k$th coordinate of $\phi(e_i,e_j)$ for $i\in[\ell],j\in[n],k\in[m]$.

Given $L\in \GL(\ell, \F_q)$ and $R\in \GL(n, \F_q)$, let $L\tA R$ be the $\ell \times n\times m$ tensor whose $k$th frontal slice is $L A_k R$. Equivalently, for a bilinear map $\phi:\F_q^\ell\times \F_q^n\rightarrow \F_q^m$, let $\phi_{L,R}(u,v):=\phi(L^\tr u, Rv)$. If $\ell=n$, let $\phi_L(u,v):=\phi_{L,L^\tr}(u,v)=\phi(L^\tr u, L^\tr v)$.

For $v\in \F_q^n$, define $\phi_v:\F_q^n\to\F_q^m$ as $\phi_v(u):=\phi(v, u)$. The degree of $v$ in $\phi$ is defined as $\deg_\phi(v):=\operatorname{rank}(\phi_v)$ because it represents the neighborhood information of $v$. 



\paragraph{Alternating matrices, alternating tensors, and alternating bilinear maps.} A matrix $A \in \M(n, \F_q)$ is \textit{skew-symmetric}, if for any $u,v \in \F_q^n$, we have $u^\tr Av = -v^\tr Au$, or equivalently $A=-A^\tr$. A matrix $A \in \M(n, \F_q)$ is \textit{alternating}, if for any $u \in \F_q^n$, we have $u^\tr Au = 0$. In other words, $A$ represents an \textit{alternating bilinear form}. Note that in characteristic $\neq 2$, alternating is the same as skew-symmetric. The linear space of $n\times n$ alternating matrices over $\F_q$ is denoted by $\Lambda(n, \F_q)$.

A tensor $\tA\in\T(n\times n\times m, \F_q)$ is \textit{frontal-alternating}, if for any $u \in \F_q^n$, we have $u^\tr \tA u = 0$, or equivalently, all of its frontal slices are alternating matrices in $\Lambda(n, \F_q)$. The linear space of $n\times n\times m$ frontal-alternating tensors over $\F_q$ is denoted by $\AT(n, m,\F_q)$.

A bilinear map $\phi:\F_q^n\times \F_q^n\rightarrow \F_q^m$ is \textit{alternating}, if for any $u \in \F_q^n$, we have $\phi(u,u)= 0$.

\paragraph{Distributions over random bilinear maps.} We denote by $\mathcal{AT}(n, m, q)$ the random distribution over all alternating bilinear maps from $\F_q^n\times \F_q^n$ to $\F_q^m$ defined in \cref{def:LinER}, 
denote by $\mathcal{AT}(n, m, q, d)$ the random distribution over all alternating bilinear maps from $\F_q^n\times \F_q^n$ to $\F_q^m$ after planting a $d$-TI space defined in \cref{def:PlLinER}, and denote by $\mathcal{AT}(n, m, q, \tB)$ the random distribution over all alternating bilinear maps from $\F_q^n\times \F_q^n$ to $\F_q^m$ after planting a given tensor $\tB$ defined in \cref{def:PlRandomLinER}.




\paragraph{Cryptographic adversary.} 
An \textit{adversary} is a probabilistic polynomial-time (PPT) algorithm.

\section{On the hardness of the planted totally-isotropic space problem}\label{sec:algorithm}

\subsection{Review of some algorithms for the planted clique problem}

\paragraph{The brute-force algorithm.} To recover the planted clique in $G$ sampled from $\ER(n, 1/2, d)$, the brute-force algorithm is to enumerate all size-$d$ vertex subsets and verify if it is a clique. If the size of the clique is more than $2\log n$, it would be the planted one with high probability. This gives a natural $\binom{n}{d}\cdot \poly(n)$-time algorithm.

\paragraph{A quasipolynomial-time algorithm.} The planted clique problem for graphs from $\ER(n, 1/2, d)$ admits a quasipolynomial-time algorithm. 
Take a graph $G$ from $\ER(n, 1/2)$ and plant in a size-$d$ clique $C$ with $d=\omega(\log n)$. Recall that the maximum clique of a random graph from $\ER(n, 1/2)$ is of size $(2+o(1))\log n$. One can then enumerate over vertex sets of size $3\log n$. Note that any subset $D$ of $C$ of size $3\log n$ is a clique, and every vertex in $C$ is connected to every vertex in $D$. From these two properties, once such $D\subseteq C$ of size $3\log n$ is found, we can recover $C$ with high probability.

\paragraph{The degree profile algorithm.}  Ku\v{c}era \cite{Ku95} observed that if $d\geq C\cdot \sqrt{n\log n}$ then the size-$d$ planted clique can be recovered in polynomial time. This is based on the fact that in this case those vertices in the planted clique are almost surely those with the largest degrees. Ku\v{c}era extended this observation and showed that it is applicable to some more complicated random models. 


\paragraph{The spectral algorithm.} In \cite{AKS98}, Alon, Krivelevich and Sudakov presented a spectral algorithm for the planted clique when $d=C\cdot\sqrt{n}$. Their method exploits the fact that a large clique acts as a low-rank perturbation to the random graph, causing the second largest eigenvector of the adjacency matrix to concentrate on the clique vertices. By computing this eigenvector and selecting vertices with large corresponding coordinates, the algorithm recovers the planted clique in polynomial time with high probability. 

\paragraph{An algorithm based on Lov\'asz $\theta$-function.} In \cite{FK03}, another algorithm that achieves recovery of the planted clique for $d=C\cdot \sqrt{n}$ was obtained using Lov\'asz $\theta$-function, with the feature that it works for more general random graph models. 

\paragraph{More recent developments.} \cite{FR10,DGP14} tackled the case of $d=C\cdot \sqrt{n}$ with combinatorial methods. These techniques rely on detailed analysis of degree and outperform the spectral algorithm in terms of time complexity and effectiveness for smaller constants $C$.

The algorithm in \cite{FR10} is called Low Degree Removal, as it removes the vertex with the lowest degree iteratively to purify the clique. The removed vertices are then re-examined in reverse order to recover any mistakenly discarded clique members. By rigorously analyzing the degree distributions and exploiting the statistical gap between clique and non-clique vertices, this method achieves a time complexity of $O(n^2)$ with constant success probability.

In \cite{DGP14}, a different iterative filtering strategy was employed to handle vertex degrees. Instead of simple elimination, the algorithm repeatedly samples a random subset of vertices and retains only those with sufficiently high connectivity to this subset. This effectively concentrates the clique density within shrinking subgraphs. It is proven to work for $C \approx 1.65$ and achieves high success probability.

\subsection{Algorithmic landscape of the planted TI space problem}

A central question in the planted TI space problem is to understand how the computational complexity depends on the planted dimension $d$. We first summarize our current algorithmic understanding, and then introduce them in the following subsections.

\begin{itemize}
    \item When $d \leq 2\log n$, because we set $m = \lceil n / \log n \rceil$, the planted structure is comparable to the largest totally-isotropic space in a random instance, which has dimension $\Theta(\log n)$ with high probability (\Cref{thm:TI-dim}). Thus, distinguishing or recovering the planted structure is impossible.

    \item $2\log n \ll d \ll n/2$ is the main regime of interest. In subsection \ref{sec:exp} we proposed an algorithm that runs in time $q^{O(n \log n)}$ based on subspace enumeration. No subexponential-time algorithms are known in this range. We conjecture that for some constant $C \geq 3$, the problem requires time $q^{\Omega(n)}$ when $d=\lceil n/C\rceil$, indicating exponential hardness.

    \item Near the critical threshold $d \approx n/2$, the rank profile method in subsection \ref{sec:distinguisher} provides a distinguisher. Specifically, the algorithm works in the regime $d = n/2 - c$ for $c=O(\sqrt{n/\log n})$, and its time complexity is polynomial in $n$ and exponential in $c$. However, this approach does not lead to a recovery algorithm directly.

    \item When $d \geq n/2$, we show in subsection \ref{sec:poly} that the planted TI space can be recovered in polynomial time using techniques from the non-commutative rank problem. This suggests that $d = n/2$ may serve as an algorithmic threshold for efficient recovery.
\end{itemize}

Overall, these results suggest the existence of a statistical-computational gap analogous to that of the planted clique problem, but potentially with a stronger hardness barrier. While planted clique is conjectured to be quasipolynomial-time hard, our evidence points toward exponential-time hardness in the planted TI space setting (\Cref{conj:search}).

\subsection{A \texorpdfstring{$q^{O(n\log n)}$}{q^{O(n log n)}}-time recovery algorithm} \label{sec:exp}


Let $\phi:\F_q^n\times\F_q^n\to\F_q^m$ be an alternating bilinear map. For a subspace $W\leq \F_q^n$, we define its orthogonal complement with respect to $\phi$ by $W_\phi^\perp:=\{u\in \F_q^n\mid \forall w\in W, \phi(u, w)=0\}$. 

\begin{fact}
    Given $W$ and $\phi$, a linear basis of $W_\phi^\perp$ can be computed in polynomial time. 
\end{fact}
\begin{proof}
    Let $\{w_1, \dots, w_r\}$ be a linear basis of $W$. Let $\phi$ be represented by $(A_1, \dots, A_m)$ where $A_i\in\Lambda(n, \F_q)$. Then $W_\phi^\perp$ can be obtained as the orthogonal complement of $\{A_iw_j\mid i\in[m], j\in[r]\}$. 
\end{proof}

\begin{proposition}\label{prop:general-algorithm}
    There exists an algorithm that recovers the planted TI space for a random $\phi$ sampled from $\PlLinER(n, m=\lceil n/\log n\rceil, q, d)$ where $d\gg (2+o(1))\log n$ with probability
\(1-q^{-\Omega(n)}\) in time \(q^{O(n\log n)}\).
\end{proposition}
\begin{proof}
    Suppose the planted TI space is $V\leq\F_q^n$. Denote $r=\lceil3\log n\rceil$. Consider the following algorithm. 
    \begin{enumerate}
    \item Enumerate all \(r\)-dimensional subspaces
    \(S\leq\F_q^n\).

    \item For each \(S\), test whether \(S\) is totally isotropic
    with respect to \(\phi\). If not, continue to the next subspace.

    \item If \(S\) is totally isotropic, compute $T=S_\phi^\perp$. If \(\dim(T)=d\) and \(T\) is totally isotropic, output \(T\) and
    terminate. Otherwise, continue to the next subspace.

    \item If no such \(T\) is found, output \(\perp\).
    \end{enumerate}
    
    The algorithm runs in time $q^{O(n\log n)}$ due to the enumeration cost in Step 1 and the fact that other steps run in time $\poly(n, \log q)$.

    For the correctness, we show that once a $(3\log n)$-dimensional subspace $S$ of the planted subspace $V$ is hit, $S^\perp_\phi=V$ with probability \(1-q^{-\Omega(n)}\).
    
    Fix a subspace \(S_0\leq V\) of dimension $\lceil 3\log n\rceil$, and choose
a decomposition
\[
\F_q^n=V\oplus U.
\]
With respect to this decomposition, each matrix representing \(\phi\)
has the form
\[
A_k=
\begin{pmatrix}
0 & B_k\\
-B_k^\tr & C_k
\end{pmatrix},
\]
where \(B_k\in\F_q^{d\times(n-d)}\) is uniformly random.

Let \(R\in\F_q^{d\times r}\) be a basis matrix of \(S_0\). Then
\[
(S_0)_\phi^\perp
=
V\oplus\ker(H),
\qquad
H=
\begin{pmatrix}
R^\tr B_1\\
\vdots\\
R^\tr B_m
\end{pmatrix}
\in\F_q^{mr\times(n-d)}.
\]
The matrix \(H\) is uniformly random. Since \(mr\ge3n\),
\[
\Pr[\rank(H)<n-d]
\leq q^{n-d-mr}
\leq q^{-2n}.
\]
Consequently,
\[
(S_0)_\phi^\perp=V
\]
with probability \(1-q^{-\Omega(n)}\). Hence the algorithm encounters
and accepts \(V\) with high probability.

It remains to show that the algorithm does not output another space.
Let \(W\ne V\) be a \(d\)-dimensional subspace and write
\[
\dim(V\cap W)=d-t.
\]
The probability that \(W\) is also totally isotropic is
\[
q^{-m\left(\binom d2-\binom{d-t}{2}\right)}
=
q^{-m t(2d-t-1)/2}.
\]
Moreover, the number of such subspaces \(W\) is at most
\[
q^{t(n-t)+O(t)}.
\]
Therefore, by a union bound and \(d\ge3\log n\),
\[
\Pr[\exists\,W\ne V:\dim(W)=d
\text{ and \(W\) is totally isotropic}]
\leq q^{-\Omega(n)}.
\]
Thus, with high probability, \(V\) is the unique \(d\)-dimensional TI
space, so the first space accepted by the algorithm must be \(V\).

\end{proof}

\subsection{A distinguisher through rank profiles} \label{sec:distinguisher}

In this section, we wish to distinguish if an alternating bilinear map is drawn from $\LinER(n, m=\lceil n/\log n\rceil, q)$ or $\LinER(n, m=\lceil n/\log n\rceil, q, d)$.
It is natural to ask whether degree-based algorithms for planted clique, such as those of \cite{Ku95}, admit an analogue for the planted TI space problem.

In a random graph, the degree of a vertex measures its connectivity to the rest of the graph. We can establish a similar metric for an alternating bilinear map $\phi:\F_q^n\times\F_q^n\to\F_q^m$. By fixing a vector $v \in \F_q^n$, we can observe its connection to other vectors through the linear map $\phi_v:\F_q^n\to\F_q^m$ defined as $\phi_v(u) = \phi(v, u)$. In the corresponding tensor $\tA\in \T(n\times n\times m, \F_q)$, $\phi_v$ can be represented by a linear combination of the lateral slices corresponding to the coordinates of $v$. Based on this structural correspondence, we define the degree of $v$ in $\phi$ as $\deg_\phi(v):=\rank(\phi_v)$.

However, a direct generalization of graph-based degree algorithms would be ineffective. In random graphs, vertex degrees follow binomial distributions that exhibit sufficient variance, creating distribution tails that algorithms can exploit even when the means are close. 
In contrast, the lateral slices in our tensor are $n\times m$ matrices, meaning $\deg_\phi(v)=\rank(\phi_v)$ is upper bounded by $m= \lceil n/\log n\rceil$. Since $n \gg m$, zeroing out $\lceil n/C\rceil$ rows in the planted structure has a negligible impact on its rank. 

Instead, we analyze the rank profile of the frontal slices and find some interesting results for the distinguishing problem. For simplicity, we assume throughout this section that $n$ is an even integer. Since the rank of any alternating matrix is always even, an odd $n$ forces the maximum rank to degenerate to $n-1$. However, the overall conclusions remain similar in the odd case.

\paragraph{When $d > n/2$.}

The regime $d > n/2$ allows for a straightforward polynomial-time distinguisher. There exists a deterministic algorithm running in $O(n^3)$ time that distinguishes $\LinER(n, \lceil n/\log n\rceil, q)$ from $\LinER(n, \lceil n/\log n\rceil, q, d)$ with success probability at least $1-O(1/q)$.

If the alternating bilinear map $\phi$ is sampled from $\LinER(n, \lceil n/\log n\rceil, q, d)$, let $S= L^\tr \tilde A L$ be a frontal slice of the corresponding tensor $\tA_{\phi}$ (Step 3 in \Cref{def:PlLinER}). Here, $L$ is an invertible basis transformation matrix so that $\rank(S)=\rank(\tilde A)$. We can decompose $\tilde A$ as the sum of two block matrices:
\[
\tilde A = \begin{pmatrix} \bzero_{d} & B \\ -B^\tr & D \end{pmatrix} 
= \begin{pmatrix} \bzero_{d} & B \\ \bzero_{(n-d) \times d} & \bzero_{n-d} \end{pmatrix} 
+ \begin{pmatrix} \bzero_{d} & \bzero_{d \times (n-d)} \\ -B^\tr & D \end{pmatrix}.
\]
The first matrix on the right-hand side has at most $n-d$ non-zero columns, and the second has at most $n-d$ non-zero rows. By the subadditivity of matrix rank,
\[
\rank(\tilde A) \leq (n-d) + (n-d) = 2n - 2d < n.
\]
In contrast, for a slice $S$ taken from $\phi \sim \LinER(n, \lceil n/\log n\rceil, q)$, the probability that it is full rank is $\prod_{j=1}^{n/2}(1-q^{1-2j}) = 1-O(1/q)$, given our assumption that $n$ is even. Therefore, the distinguishing algorithm simply computes the rank of an arbitrary frontal slice: if $\rank(S) < n$, it identifies the map as drawn from the planted distribution; otherwise, it outputs random.



\paragraph{When $d = n/2$.}

At this critical threshold, the deterministic rank-deficiency test on a single slice no longer applies. Nevertheless, by conducting a more refined analysis of the matrix rank profile, we can still obtain a distinguisher for random tensors and planted ones. First, we introduce the following lemmas, summarizing asymptotic rank distributions of random matrices over finite fields.

\begin{lemma}[{\cite{S88, F02}}]\label{lemma:rand-corank}
For any integers $n \geq 1$ and $0 \leq c \leq n$, if $A$ is a uniformly random matrix in $\M(n, \F_q)$, then
\[
\Pr[\operatorname{corank}(A)=c]\sim q^{-c^2}
\qquad\text{as } q\to\infty.
\]
\end{lemma}

\begin{lemma}[{\cite{S88, F02}}]\label{lemma:alt-corank}
For any integers $n \geq 1$ and $0 \leq c \leq n$ such that $n-c$ is even, if $A$ is a uniformly random alternating matrix in $\Lambda(n, \F_q)$, then
\[
\Pr[\operatorname{corank}(A)=c]\sim q^{\frac{-c^2+c}{2}}
\qquad\text{as } q\to\infty.
\]
\end{lemma}

\begin{lemma}[{\cite{S88, F02}}]\label{lemma:rec-corank}
For any integers $n, l, c$ such that $0 \leq l-c \leq l \leq n$, if $A$ is a uniformly random matrix in $\mathbb{F}_q^{l \times n}$, then
$$\Pr[\rank(A) = l-c] \sim q^{- c^2-(n-l)c} \qquad \text{as } q \to \infty.$$
\end{lemma}

\begin{algorithm}[H]
\caption{A distinguisher for the random tensor and tensor with a planted TI space when $d=n/2$}
\label{alg:distinguisher}
    \textbf{Input:} An alternating bilinear map $\phi: \mathbb{F}_q^n \times \mathbb{F}_q^n \to \mathbb{F}_q^m$ represented by a tuple of alternating matrices $A = (A_1, A_2, \dots, A_m) \in \Lambda(n,\F_q)^m$, sampled from either $\LinER(n, m=\lceil n/\log n\rceil, q)$ or $\LinER(n, m=\lceil n/\log n\rceil, q, d=n/2)$ \\
    \textbf{Output:} ``Random'' or ``Planted''.
    \begin{enumerate}
    \item Let $N = \lceil q^{5.5} \rceil$. Choose $N$ pairwise non-collinear nonzero coefficient vectors $c^{(k)} = (c_1^{(k)}, \dots, c_m^{(k)}) \in \mathbb{F}_q^m$ for $k \in [N]$.
    \item For each vector $c^{(k)}$, compute the corresponding linear combination $M_k = \sum_{i=1}^m c_i^{(k)} A_i$.
    \item Compute the rank of each matrix $M_k$ using Gaussian elimination.
    \item If there exists at least one $k \in [N]$ such that $\rank(M_k) = n-4$, output ``Planted''. Otherwise, output ``Random''.
\end{enumerate}
\end{algorithm}

Note that Step 1 requires $n \geq 30$ to ensure that the number of available non-collinear directions should satisfy $\frac{q^m-1}{q-1} \geq q^{5.5}$.

\begin{theorem} \label{thm:n/2}
Let $n \geq 30$ be an even integer and $m = \lceil n/\log n \rceil$. \Cref{alg:distinguisher} distinguishes an alternating bilinear map sampled from the planted model $\LinER(n, m, q, d=n/2)$ from one sampled from the random model $\LinER(n, m, q)$ with an advantage of at least $1 - O(q^{-0.5})$. Furthermore, the algorithm runs in $O(q^{5.5} \cdot n^3)$ time.
\end{theorem}
\begin{proof}
Let
\[
X_k:=\mathbf 1_{\{\rank(M_k)=n-4\}},
\qquad
X:=\sum_{k=1}^N X_k.
\]
The decision rule of the distinguishing algorithm is defined as follows: if $X \geq 1$, it outputs ``Planted''; if $X = 0$, it outputs ``Random''.

If $\phi$ is sampled from $\LinER(n, m, q)$, any non-zero linear combination of $A_i$ yields a uniformly distributed random alternating matrix $M$ in $\Lambda(n,\F_q)$. By \Cref{lemma:alt-corank}, the probability that $\rank(M)=n-4$ is
$$\Pr_{\phi\sim \LinER(n, m, q)}[\rank(M) = n-4] = \Theta(q^{-6}).$$
By the union bound, the probability of a false positive is bounded by
$$\Pr_{\phi\sim \LinER(n, m, q)}[X \geq 1] \leq \sum_{k=1}^N \Pr_{\phi\sim \LinER(n, m, q)}[\rank(M_k) = n-4] = N \cdot \Theta(q^{-6}) = O(q^{-0.5}).$$

If $\phi$ is sampled from $\LinER(n, m, q, d=n/2)$,
a random linear combination $M$ is congruent to a block matrix of the form $\begin{pmatrix} \bzero_d & B \\ -B^\tr & D \end{pmatrix}$, where $B, D$ are uniformly random matrices with dimension $n/2$, and $D$ is alternating.

To achieve $\rank(M) = n-4$, the dominant case requires $\rank(B) =  n/2 - 2$. In this case, there exist invertible matrices $P, Q$ such that $P^\tr B Q = \operatorname{diag}(I_{d-2}, \bzero_{2 \times 2})$. Applying the congruence transformation $T = \operatorname{diag}(P, Q)$ to $M$ yields
$$M' = T^\tr M T = \begin{pmatrix} \bzero & P^\tr B Q \\ -Q^\tr B^\tr P & Q^\tr D Q \end{pmatrix}.$$
Let $D' = Q^\tr D Q$. Since $D$ is a uniformly random alternating matrix and $Q$ is invertible, $D'$ remains uniformly random. We partition $D'$ conformably with $P^\tr B Q$. By applying block Gaussian elimination using the $I_{d-2}$ and $-I_{d-2}$ blocks, we get
$$M' = \begin{pmatrix} \bzero & \bzero & I_{d-2} & \bzero \\ \bzero & \bzero & \bzero & \bzero \\ -I_{d-2} & \bzero & D'_{11} & D'_{12} \\ \bzero & \bzero & -(D'_{12})^\tr & F \end{pmatrix} \to M'' = \begin{pmatrix} \bzero & \bzero & I_{d-2} & \bzero \\ \bzero & \bzero & \bzero & \bzero \\ -I_{d-2} & \bzero & \bzero & \bzero \\ \bzero & \bzero & \bzero & F \end{pmatrix},$$
where $F$ is a $2 \times 2$ random alternating matrix.

From this form, we know that $\rank(M) = \rank(M'') = 2(d-2) + \rank(F) = (n-4) + \rank(F)$. Therefore, for the overall rank to be exactly $n-4$, the $2 \times 2$ random alternating matrix $F$ must contribute a rank of $0$, whose probability is exactly $q^{-1}$ by \Cref{lemma:alt-corank}. By \Cref{lemma:rand-corank}, the probability that a uniformly random matrix $B$ has corank $2$ is $\Theta(q^{-4})$. Since $B$ and $D$ are sampled independently, the joint probability is $\Theta(q^{-5})$. Other configurations that may also lead to $\rank(M)=n-4$ contribute only lower-order terms, hence the total probability is still dominated by the case $\operatorname{corank}(B)=2$. Therefore, we denote
\[
p_1:=\Pr_{\phi\sim \LinER(n, m, q, n/2)}[\rank(M_k)=n-4]=\Theta(q^{-5}).
\]
Then
\[
\E_{\phi\sim \LinER(n, m, q, n/2)}[X]=Np_1=\Theta(q^{5.5})\Theta(q^{-5})=\Theta(q^{0.5}).
\]


The requirement that the coefficient vectors be pairwise non-collinear ensures that for any distinct $k,\ell\in[N]$, the vectors $c^{(k)}$ and $c^{(\ell)}$ are linearly independent. In the planted model, conditioned on the hiding matrix \(L\), the corresponding linear combinations $M_k$ and $M_\ell$ are independent planted slices. Since the event $\rank(M_k)=n-4$ is invariant under congruence, its conditional probability does not depend on the hidden basis choice. Hence the indicators $X_k=\mathbf 1_{\{\rank(M_k)=n-4\}}$ are pairwise independent, so the covariance terms vanish and
\[
\Var_{\phi\sim \LinER(n, m, q, n/2)}[X]
=
\sum_{k=1}^N \Var[X_k].
\]

Moreover, each $X_k$ is a Bernoulli random variable with parameter $p_1$, so
\[
\Var[X_k]=p_1(1-p_1)\leq p_1.
\]
Thus,
\[
\Var_{\phi\sim \LinER(n, m, q, n/2)}[X]
\le
\sum_{k=1}^N p_1
=
\E_{\phi\sim \LinER(n, m, q, n/2)}[X]
=
\Theta(q^{0.5}).
\]

Now observe that the algorithm fails under the planted model only if $X=0$. In that case,
\[
|X-\E_{\phi\sim \LinER(n, m, q, n/2)}[X]|
=
\E_{\phi\sim \LinER(n, m, q, n/2)}[X].
\]
By Chebyshev's inequality,
\[
\Pr_{\phi\sim \LinER(n, m, q, n/2)}[X=0]
\le\Pr\!\left[\left|X-\E[X]\right|\geq \E[X]\right]
\le\frac{\Var[X]}{\E[X]^2}
\leq \frac{1}{\E[X]}
=O(q^{-0.5}).
\]

Therefore, the distinguishing advantage of the algorithm can be strictly lower-bounded as follows:
$$\begin{aligned}
\text{Adv} &= \left| \Pr_{\phi\sim \LinER(n, m, q, n/2)}[X \geq 1] - \Pr_{\phi\sim \LinER(n, m, q)}[X \geq 1] \right| \\
&= \left| (1 - \Pr_{\phi\sim \LinER(n, m, q, n/2)}[X = 0]) - \Pr_{\phi\sim \LinER(n, m, q)}[X \geq 1] \right| \\
&\geq 1 - O(q^{-0.5}) - O(q^{-0.5}) \\
&\geq 1 - O(q^{-0.5}).
\end{aligned}$$

The time complexity of the algorithm is dominated by the rank computation of $N$ matrices. Gaussian elimination for an $n \times n$ matrix takes $O(n^3)$ operations. Thus, the total time complexity is $O(q^{5.5} \cdot n^3)$.
\end{proof}

\paragraph{When $d = n/2 - c$.}
For $d = n/2 - c$, we analyze the event $\operatorname{corank}(M)=2r$ for a random nonzero linear combination $M$ of the slices. By \Cref{lemma:alt-corank}, in the random model, the probability of $\operatorname{corank}(M)=2r$ is governed by 
\begin{equation}\label{eq:random-minus-c}
    q^{\frac{-(2r)^2+2r}{2}}=q^{-2r^2+r}.
\end{equation}

In the planted model, $M$ is congruent to a block matrix of the form $\begin{pmatrix} \bzero_{d} & B_{d\times(d+2c)} \\ -B^\tr_{(d+2c)\times d} & D_{(d+2c)\times (d+2c)} \end{pmatrix}$.
Let $s=d-\rank(B) \geq 0$, then $B$ can be transformed into $\begin{pmatrix} I_{d-s} & \bzero \\ \bzero & \bzero \end{pmatrix}$. Consequently, applying block Gaussian elimination leaves a $(2c+s)\times(2c+s)$ alternating matrix $F$. 
$$M' = \begin{pmatrix} \bzero & \bzero & I_{d-s} & \bzero \\ \bzero & \bzero & \bzero & \bzero \\ -I_{d-s} & \bzero & D'_{11} & D'_{12} \\ \bzero & \bzero & -(D'_{12})^\tr & F \end{pmatrix} \to M'' = \begin{pmatrix} \bzero & \bzero & I_{d-s} & \bzero \\ \bzero & \bzero & \bzero & \bzero \\ -I_{d-s} & \bzero & \bzero & \bzero \\ \bzero & \bzero & \bzero & F \end{pmatrix}.$$

Since $\rank(B)=d-s$, in order to get $\rank(M)=n-2r$, we need $\rank(F)=n-2r-2(d-s)=2c-2r+2s$, hence $\operatorname{corank}(F)=2r-s$. 

By \Cref{lemma:alt-corank} and \Cref{lemma:rec-corank}, 
\begin{align*}
\Pr[\operatorname{corank}(F)=2r-s] &\sim q^{\frac{-(2r-s)^2+(2r-s)}{2}}, \\
\Pr[\rank(B)=d-s] &\sim q^{-s(2c+s)}.
\end{align*}
Therefore, the probability that $\operatorname{corank}(M)=2r$ is dominated by the largest value of the following $g(s)$ over all choices of $s$:
\begin{equation}\label{eq:plant-minus-c}
    g(s)=\frac{-(2r-s)^2+(2r-s)}{2}-s(2c+s)=-\frac{3}{2}s^2+(2r-2c-\frac{1}{2})s-2r^2+r.
\end{equation}
Note that there is a constraint $\operatorname{corank}(F) \leq \dim(F)$, i.e. $s\geq r-c$. Also $s\geq 0$.
Comparing Equations~\eqref{eq:random-minus-c} and~\eqref{eq:plant-minus-c}, we obtain the following general pattern:
\begin{itemize}
    \item If $r\leq c+1$, the leading exponent is $g(s)=-2r^2+r$. Therefore,
    \[
    \Pr_{\phi\sim \LinER(n, m, q, n/2-c)}[\operatorname{corank}(M)=2r]
    \asymp
    \Pr_{\phi\sim \LinER(n, m, q)}[\operatorname{corank}(M)=2r],
    \]
    so the event $\operatorname{corank}(M)=2r$ does not provide a useful distinguisher at the level of leading asymptotics.

    \item If $r\geq c+2$, the dominant term is achieved at $s = r - c$. In particular, the first corank level at which an asymptotic separation appears is $r=c+2$, i.e. $\operatorname{corank}(M)=2c+4$. Therefore, $s=2$ and $\max g(s)=-2c^2 - 7c - 5$.
    
    In this case, $\Pr_{\phi\sim \LinER(n, m, q)}[\operatorname{corank}(M)=2c+4] \sim q^{-2c^2 - 7c - 6}$.
    Then the planted probability becomes asymptotically larger than the random one.

\end{itemize}

Therefore, we obtain a similar result to \Cref{alg:distinguisher} and \Cref{thm:n/2}:

\begin{theorem} \label{thm:general_c}
Let $c \geq 0$ be an integer satisfying
\[
2c^2+7c+\frac{13}{2}\leq \left\lceil \frac{n}{\log n}\right\rceil.
\]
There exists an algorithm that distinguishes an alternating bilinear map sampled from the planted model $\LinER(n, m=\lceil n/\log n\rceil, q, d=n/2-c)$ from one sampled from the random model $\LinER(n, m, q)$ with an advantage of at least $1 - O(q^{-0.5})$. Moreover, the algorithm runs in $O(q^{2c^2+7c+\frac{11}{2}} \cdot n^3)$ time.
\end{theorem}

\begin{proof}
As demonstrated by the analysis of $g(s)$, the first corank level at which an asymptotic separation appears is $2c+4$. Note that this general conclusion aligns with the case of $c=0$, i.e. $d=n/2$.

We use the same sampling strategy as in the case $d=n/2$, namely, testing many pairwise non-collinear linear combinations and checking whether at least one of them has this prescribed corank. Then the required number of samples is $q^{\,2c^2+7c+\frac{11}{2}}$.
Accordingly, the distinguishing algorithm runs in time
\[
O\!\left(q^{\,2c^2+7c+\frac{11}{2}}\cdot n^3\right).
\]
Thus, for constant $c$, higher-corank events still yield a distinguisher, but the required complexity grows rapidly as $c$ increases.

Moreover, the total number of available pairwise non-collinear coefficient vectors is at most
$$\frac{q^m-1}{q-1}\asymp q^{m-1},
\qquad \text{where } m=\left\lceil \frac{n}{\log n}\right\rceil.$$
so in order for the sampling-based distinguisher to be feasible, we must at least have
$$2c^2+7c+\frac{13}{2} \leq \left\lceil \frac{n}{\log n}\right\rceil.$$

Following a similar proof to the $d=n/2$ case, the advantage of the distinguisher is also $1-O(q^{-1/2})$.
\end{proof}


\subsection{A polynomial-time recovery algorithm for \texorpdfstring{$d \geq n/2$}{d >= n/2}} \label{sec:poly}





When the planted TI space dimension is $d\geq n/2$, we have the following.

\begin{proposition}\label{prop:n-over-2}
There exists a polynomial-time algorithm that recovers the planted TI space for alternating bilinear maps sampled from $\LinER(n, \lceil n/\log n\rceil, q, d)$, $d\geq n/2$, with high probability.
\end{proposition}
\begin{proof}
	Let $\phi:\F_q^n\times\F_q^n\to\F_q^m$ be an alternating bilinear map sampled from $\LinER(n, m=\lceil n/\log n\rceil, q, d)$ with $d\geq n/2$. Let $V\in\Gr(d, \F_q^n)$ be the planted TI space in $\phi$. That is, $\phi(V, V)$ is the zero bilinear map. 

    Instead of searching for $V\in\Gr(d, \F_q^n)$ such that $\phi(V, V)=0$, we shall expand our search range to $U, W\in\Gr(d, \F_q^n)\times\Gr(d, \F_q^n)$, such that $\phi(U, W)=0$. 
    


    \begin{claim}\label{claim:TI-pair}
        Let $\phi:\F_q^n\times\F_q^n\to\F_q^m$ be sampled from $\LinER(n, m=\lceil n/\log n\rceil, q, d\gg 2\log n)$, with $V\in\Gr(d, \F_q^n)$ as the planted TI space. Then with high probability, $(V, V)$ is the only $(U, W)\in\Gr(d, \F_q^n)\times\Gr(d, \F_q^n)$ such that $\phi(U, W)=0$.
    \end{claim}
    \begin{proof}
        To illustrate the idea, let us first consider $(U, V)$ where $U\in \Gr(d, \F_q^n)$ and $V$ is the planted TI space. We further assume that $\dim(U\cap V)=d-1$. Then there exists a non-zero $u\in U\setminus V$. In order to satisfy $\phi(U, V)=0$, we need to ensure that $\phi(u, V)=0$, and this probability (over $\phi$ sampled from $\LinER(n, m, q, d)$) is $1/q^{dm}$. As $m=\lceil n/\log n\rceil$ and $d\gg 2\log n$, $1/q^{dm}=1/q^{\omega(n)}$. On the other hand, the number of possible $u\in U\setminus V$ is upper bounded by $q^n$. Then by a union bound, 
        $$
        \Pr[\exists U\in\Gr(d, \F_q^n)\mid \dim(U\cap V)=d-1, \phi(U, V)=0]\leq q^n\cdot \frac{1}{q^{dm}}=\frac{q^n}{q^{\omega(n)}}=o(1).
        $$

        Similarly, consider $(U, V)$ where $U\in \Gr(d, \F_q^n)$ such that $V\cap U$ is of codimension $k$ in $U$. For a fixed $U$ as such, the probability of $\phi(U, V)=0$ is upper bounded by $1/q^{kdm}$. The number of $k$-dimensional subspaces in $\F_q^n/U$ is upper bounded by $q^{dn}$. As $m=\lceil n/\log n\rceil$ and $d=\omega(\log n)$, we have $\Pr[\exists U\in\Gr(d, \F_q^n)\mid \dim(U\cap V)=d-k, \phi(U, V)=0]\leq o(1)$ as above.

        The case of arbitrary $(U, W)\in \Gr(d, \F_q^n)\times\Gr(d, \F_q^n)$ satisfying $\phi(U, W)=0$ can be dealt with in a similar way, by considering their intersection with $V$. We omit the routine calculations here.
    \end{proof}

    Given Claim~\ref{claim:TI-pair}, to recover the planted TI space $V\in\Gr(d, \F_q^n)$ when $d\geq n/2$, we can instead find $(U, W)\in\Gr(d, \F_q^n)\times\Gr(d, \F_q^n)$ such that $\phi(U, W)=0$. 

    When $d>n/2$, $U$ in such $(U, W)$ satisfies that $\dim(\linspan\{\cup_{i\in[m]}A_i(U)\})<\dim(U)$. Indeed, $\linspan\{\cup_{i\in[m]}A_i(U)\}$ is contained in a complement subspace of $W$, so $\dim(U)=d>n-d\geq
\dim\!\left(
\linspan\{\cup_{i\in[m]}A_i(U)\}
\right)$. Such a subspace is called a shrunk subspace. The shrunk subspace problem received considerable attention in algebraic complexity, and it admits three deterministic polynomial-time solutions in \cite{GGOW20,IQS18,HH21}. In particular, the algorithms in \cite{IQS18,HH21} work over finite fields and output the ``canonical'' shrunk subspace, which in our setting is the planted TI space $V$ with high probability.

    When $d=n/2$, such an $(U, W)$ pair is called an invariant subspace pair in \cite{IQ23}, which presented an algorithm that computes such an invariant subspace pair, which in our setting is $(V, V)$, $V$ the planted TI space, with high probability.
\end{proof}

\subsection{Algorithms based on solving polynomial equations}\label{subsec:grobner}
It is natural to formulate the planted totally-isotropic space problem as solving a system of polynomial equations over $\F_q$. Furthermore, as the planted condition $\phi(V,V)=0$ is quadratic in the coordinates of a basis of $V$, the search problem can be formulated as a system of quadratic equations.

Let $\phi: \F_q^n \times \F_q^n \to \F_q^m$ be represented by an alternating matrix tuple $(A_1,\dots,A_m)$ for $A_k \in \Lambda(n,\F_q)$, so that $\phi(u,v)=(u^\tr A_1v,\dots,u^\tr A_mv)^\tr$. A $d$-dimensional subspace $V\leq \F_q^n$ is totally-isotropic if and only if there exists a full-row-rank matrix $X\in \GL(d\times n,\F_q)$ whose row space is $V$ such that
\[XA_kX^\tr=\mathbf{0}_{d\times d} \quad \text{for all } k\in [m].\]
Since each $A_k$ is alternating, each matrix $XA_kX^\tr$ is alternating, so it is enough to impose the equations on the entries above the diagonal. This yields a system of $m\cdot\binom{d}{2}$ quadratic equations in the entries of $X$.

To ensure that the recovered matrix $X$ has full row rank, one approach is to fix a specific $d\times d$ block of $X$ and impose it to be invertible. 
While this may leave out those invertible matrices whose corresponding blocks are singular, it is a convenient assumption that usually suffices for practical computations. For example, let $U$ be the top-left $d\times d$ block of $X$. Instead of only requiring $U$ to be nonsingular using determinant, we introduce another variable matrix $U_{\mathrm{inv}}$ intended to represent its inverse, and add the equations
\[UU_{\mathrm{inv}}=I_d.\]
Thus, the entries of both $X$ and $U_{\mathrm{inv}}$ are treated as polynomial variables. This enforces the invertibility of the chosen block $U$ inside the polynomial system itself and excludes obvious degenerate solutions.

In the planted instance, this additional constraint is compatible with the true solution with a reasonable probability. Indeed, in our generation procedure, we sample the hiding matrix $L$ uniformly at random from $\GL(n, q)$, so the top-left $d \times d$ block of $L^{-\tr}$ is invertible with probability $1-O(1/q)$. If this holds, the intended solution $X$ satisfies not only the planted tensor equations $X  A_k X^\tr=\mathbf{0}_{d \times d}$ for all $k \in [m]$, but also the block-invertibility equations through the choice of the top-left $d\times d$ block of $L^{-\tr}$. This justifies setting $U_{\mathrm{inv}}$.

Therefore, the polynomial system used in our implementation consists of two parts: the planted tensor equations $X  A_k X^\tr=0_{d \times d}$ and the auxiliary block-invertibility equations $UU_{\mathrm{inv}}=I_d$, 
where $U$ is the first $d\times d$ submatrix of $X$. We construct a polynomial system in the entries of $X$ and $U_{\mathrm{inv}}$, and then solve this polynomial system via computing the Gr\"obner basis.

This formulation is tailored to the concrete implementation. Its purpose is to model, as directly as possible, the task of recovering the planted block from the hidden public tensor under the alternating tensor planting distribution. In particular, it applies uniformly both to the case of a fixed planted alternating tensor $\tB$ and to the planted zero-block case $\tB=0$.

\begin{algorithm}[H]
\caption{Polynomial-system solving for planted alternating tensors}
    \textbf{Input:} Integers $n,d,m$, a field $\F_q$, and a public alternating tensor $\vA=(A_1,\ldots,A_m)\in \Lambda(n,\F_q)^m$. A planted tensor $\vB=(B_1,\ldots,B_m) \in \Lambda(d,\F_q)^m$, where possibly $\vB=0$.\\
    \textbf{Output:} A rank-$d$ matrix $S\in \M(d\times n, \F_q)$ such that $S\vA S^\tr=\vB$, if such $S$ exists.
    \begin{enumerate}
    \item Construct a polynomial ring $R$ over $\F_q$ in $(dn+d^2)$ variables. 
    \item Set a variable matrix $X\in \M(d\times n, R)$ using the first $dn$ variables. Set $U \leftarrow X_{[:,1..d]} \in \M(d, R)$, i.e., the submatrix of $X$ consisting of its first $d$ columns.
    \item Introduce another variable matrix $U_{\mathrm{inv}}\in \M(d, R)$ using the last $d^2$ variables.
    \item Initialize an empty equation set $\mathcal{E}$.
    
    \item For each $k\in[m]$, compute
    \[
    M_k \leftarrow X A_k X^\tr - B_k,
    \]
    and add all equations $(M_k)_{i,j}=0$ for $1\leq i,j\leq d$ to $\mathcal{E}$.
    
    \item Add all equations
    $(UU_{\mathrm{inv}}-I_d)_{i,j}=0$ for $ 1\leq i,j\leq d$
    to $\mathcal{E}$.
    
    \item Form the ideal $J=\langle \mathcal{E}\rangle \subseteq R$ and compute a Gr\"obner basis of $J$.
    \item Solve the resulting polynomial system for the entries of $X$ and $U_{\mathrm{inv}}$. 
    \item If the solution set is non-empty, return a solution $S$. Otherwise, return $\bot$.
    
\end{enumerate}
\end{algorithm}

\paragraph{Hybrid guessing–algebraic attacks.} A natural variant of the polynomial-system approach is to combine \textit{guessing} with Gr\"obner basis computation. The guess of the entries of $X$ can be interpreted as revealing partial information about the hidden transformation. Specifically, after fixing $h$ entries of $X$, the resulting polynomial system contains $dn+d^2-h$ unknowns left. This reduces the effective number of unknowns and thereby simplifies the Gr\"obner basis computation. Nevertheless, we need to guess the correct values of these entries exhaustively, so the corresponding attack incurs an additional factor $q^h$. In this way, we get a simple hybrid strategy between exhaustive guessing and a pure algebraic approach.


\paragraph{Implementation and preliminary timings.}
We implemented the above algorithm in \textsc{Magma} ~\cite{Magma} and carried out several preliminary experiments on a laptop with the following hardware: \emph{MacBook Pro (13-inch, 2019, Four Thunderbolt 3 ports), 2.4 GHz Quad-Core Intel Core i5}. Our experiments show a clear contrast between the direct and hybrid strategies. Without fixing any entries of $X$, even very small instances are already costly: for a planted instance over the prime field $\F_{131}$ with parameters $d=3$ and $n=6$, recovering the planted structure took about 41 seconds. On the same machine, we were unable to solve larger instances reliably within our computational budget.

    By contrast, after fixing a substantial number of entries of $X$ to their correct values, the residual polynomial systems for moderate parameter sizes became tractable. After fixing $n(\log n)^2$ many entries over $\F_{131}$, for $d=20$ and $n=60$, the residual system was solved in about 23 seconds; for $d=30$ and $n=90$, the running time was about 428 seconds. These timings measure only the Gröbner basis computation after the correct values of the $n(\log n)^2$ entries have been fixed. If these values are not provided as side information but instead have to be obtained by exhaustive guessing, then enumerating the $n(\log n)^2$ entries incurs an additional factor of $q^{n(\log n)^2}$. Hence, after accounting for the guessing cost, this hybrid strategy does not improve on the $q^{O(n\log n)}$-time algorithm analyzed in \cref{sec:exp}.


\section{From planting totally-isotropic spaces to planting subtensors}\label{sec:assumptions}

Throughout this section, we keep the notation from Section~\ref{sec:prelim}: we view tensors and bilinear maps via the correspondence between $\tA$ and $\phi_{\tA}$, and for the alternating setting we use the congruence action $\tA \to L^\tr\tA L$ which corresponds to $\phi_\tA(u,v)\to \phi_\tA(Lu, Lv)$.

\paragraph{The planted alternating tensor (PAT) assumption.} To accommodate cryptographic applications, we generalize the planted totally-isotropic space model
from Definition~\ref{def:PlLinER}, in which the planted subtensor is the all-zero tensor, to a more general model where an arbitrary small alternating tensor is planted into a larger ambient alternating tensor. This gives us Definition~\ref{def:PlRandomLinER}. As explained after Definition~\ref{def:PlRandomLinER}, we expect that planting in TI spaces is easier than planting in other tensors, just as the general belief that the planted cliques are the easiest to detect \cite{ABI+23}. This leads us to formulate the planted alternating tensor (PAT) assumption, which states that a random alternating tensor $\tA$ containing a planted alternating subtensor $\tB$ 
remains computationally indistinguishable from a random alternating tensor of the same size.
Formally, for an appropriate choice of parameters $n, m, q, T, \varepsilon$, we say that the ($n,m,q,T,\varepsilon,\mathtt{B}$)-PAT assumption holds if,  for every non-uniform $T(q,n)$-time adversary $\mathcal{A}$,
$$
\left|\Pr[\mathcal{A}(\tA,\tB)=1 \mid \tA \smpl \mathcal{AT}(n, m, q, \tB)]-\Pr[\mathcal{A}(\tA,\tB)=1 \mid \tA \smpl \mathcal{AT}(n, m, q)]\right| \leq \varepsilon.
$$

\paragraph{The planted alternating tensor with 2-hints (PATw2H) assumption.} The PATw2H
assumption is a strengthening 
of the PAT assumption in which the adversary is additionally provided with \emph{hints} in the form of the two vectors $(\hat{x},\hat{y})$ that relate to the hiding action in the planting procedure.
Formally, let $J_d=
\begin{bmatrix}
I_d\\
0
\end{bmatrix}
\in \F_q^{n\times d}
$ 
 and given $\tB\in\AT(d,m,\F_q)$. Recall that we defined the distribution $\PPlLinER$ on pairs of alternating bilinear maps and invertible matrices in Definition~\ref{def:PlRandomLinER}. We say the $(n,m,q,T,\varepsilon,\tB)$-PATw2H assumption holds, if for every pair of linearly independent $x,y\in\F_q^d$, the following distributions are $\varepsilon$-computationally indistinguishable for every non-uniform $T(q,n)$-time adversary:



$$
\left\{\left(\tA, \tB,\hat{x},\hat{y}\right)\left\lvert\, \begin{array}{l}
(\tA,L)\smpl \mathcal{PAT}(n,m,q,\tB) \\
\hat{x}\leftarrow L^{-1}J_dx, \hat{y}\leftarrow L^{-1}J_d y
\end{array}\right.\right\},
$$
and 
$$
\left\{\left(\tA, \tB,\hat{x},\hat{y}\right)\left\lvert\, \begin{array}{l}
\hat{x},\hat{y} \smpl \F_q^n\\
\forall k\in [m]: M_k\smpl \{M\in \Lambda(n, \F_q): \hat{x}^{\tr}M\hat{y}=x^\tr B_ky\}\\
\tA\gets (M_1,\dots,M_m)
\end{array}\right.\right\}.
$$


Similarly to \cref{conj:decision}, for the purpose of our PSM construction, we conjecture the following:
\begin{conjecture}[PATw2H Conjecture]\label{conj:\PRSwH}
There exist constants $0<c<1$, $C\geq 3$ and $d=\lceil n/C\rceil$, such that the $(n,m=\lceil n/\log n\rceil,q,q^{o(n)},\varepsilon,\tB)$-PATw2H assumption holds with $\varepsilon=c$ against all non-uniform $q^{o(n)}$-time adversaries.
\end{conjecture}
\begin{lemma}\label{lemma:bridge-sim}
    Let $\tB\in \AT(d,m,\F_q)$ be fixed. Let $z_\tB(x,y)=(x^\tr B_1 y,\dots, x^\tr B_m y)\in \F_q^m$. Assume that the $(n,m,q,T,\varepsilon,\tB)$-PATw2H assumption holds. Then there exists a non-uniform PPT simulator $\alg{Sim}_\alg{PATw2H}$ such that for every non-uniform $T(n)$-time adversary $\mathcal A$ and for every pair of linearly independent $x,y\in\F_q^d$,
\[
\left|
\Pr[\mathcal A(\tA,\hat{x},\hat{y},\tB)=1]-
\Pr[\mathcal A(\alg{Sim}_\alg{PATw2H}(\mathbbm{1}^n,\tB,z_\tB(x,y)))=1]
\right|
\leq \varepsilon(n),
\]
where $(\tA,L)\xleftarrow{\$}\mathcal{PAT}(n,m,q,\tB), \hat{x}=L^{-1}J_dx, \hat{y}=L^{-1}J_dy$.
\end{lemma}
\begin{proof}
    Define $\alg{Sim}_\alg{PATw2H}(1^n,\tB,z)$ to sample from the second distribution in the PATw2H assumption for the fixed planted tensor $\tB$ and the target output vector $z$. The claim is then exactly the fixed $\tB$ special case of the PATw2H assumption.
\end{proof}

\paragraph{The planted random alternating tensor with 2-hints (\PRATwTwoH) assumption.} The \PRATwTwoH
assumption is a slight variant of the PATw2H assumption where the adversary does not get a deterministically fixed tensor $\tB$ but gets a random tensor $\tB\smpl\mathcal{AT}(d,m,q)$. Thus, this assumption is weaker than the PATw2H assumption. Formally, recall that $J_d=
\begin{bmatrix}
I_d\\
0
\end{bmatrix}
\in \F_q^{n\times d}
$. We say the $(n,m,d,q,T,\varepsilon)$-\PRATwTwoH assumption holds, if for every pair of linearly independent $x,y\in\F_q^d$, the following distributions are $\varepsilon$-computationally indistinguishable for every non-uniform $T(q,n)$-time adversary:

$$
\left\{\left(\tA, \tB,\hat{x},\hat{y}\right)\left\lvert\, \begin{array}{l}
\tB\smpl\mathcal{AT}(d,m,q)\\
(\tA,L)\smpl \mathcal{PAT}(n,m,q,\tB) \\
\hat{x}\leftarrow L^{-1}J_dx, \hat{y}\leftarrow L^{-1}J_d y
\end{array}\right.\right\},
$$
and 
$$
\left\{\left(\tA, \tB,\hat{x},\hat{y}\right)\left\lvert\, \begin{array}{l}
\tB\smpl\mathcal{AT}(d,m,q)\\
\hat{x},\hat{y} \smpl \F_q^n\\
\forall k\in [m]: M_k\smpl \{M\in \Lambda(n, \F_q): \hat{x}^{\tr}M\hat{y}=x^\tr B_ky\}\\
\tA\gets (M_1,\dots,M_m)
\end{array}\right.\right\}.
$$

This assumption is designated for our FGSS construction, and so we conjecture the following:
\begin{conjecture}[\PRATwTwoH Conjecture]\label{conj:\PRATwTwoH}
There exist constants $0<c<1$, $C \geq 3$ and $d=\lceil n/C\rceil$, such that the $(n,m=\lceil n/\log n\rceil,d,q,q^{o(n)},\varepsilon)$-\PRATwTwoH assumption holds with $\varepsilon=c$ against all non-uniform $q^{o(n)}$-time adversaries.
\end{conjecture}

\section{Private Simultaneous Messages from Planted Tensors}\label{sec:PSM}

In this section, we present a computational 2-party private simultaneous messages (PSM) protocol with public information, based on the alternating planting framework from Section~\ref{sec:assumptions}. 
We construct the protocol for alternating bilinear functions $f: \F_q^n \times \F_q^n \to \F_q^m$, and it can be extended to arbitrary bilinear functions via a simple alternating lift. This goes beyond prior constructions based on planted graph assumptions, which were primarily designed for Boolean functions of the form $[n] \times [n] \to \{0,1\}$.
\subsection{PSM protocol with public information}


A private simultaneous messages (PSM) protocol is a cryptographic tool that enables two parties to independently encode their initial secret inputs with randomness. The encodings are then sent non-interactively, and from these encodings one can evaluate a function $f$ and output the same value as applying both parties' initial secret inputs to $f$. Crucially, any observer who only has access to the encodings gains no knowledge about the secret inputs beyond what is revealed by the function output itself. Building upon this standard model, \cite{ABIKN23} proposed a computational variant with public information in which a setup phase provides the parties with secret shared randomness as well as some public information. This public information is essential for reconstructing the output of the function, but it does not reveal anything further about the secret inputs under certain computational assumptions. The model can be formalized as follows, with a slight change from Boolean functions to bilinear maps.

\begin{definition}[PSM protocols with public information \cite{ABIKN23} for bilinear maps]
Let $\mathcal{F}:=\{f_n\}_{n \in \mathbb{N}}$ be a family of bilinear maps such that $f_n:\F_q^n \times\F_q^n \rightarrow\F_q^m$.\footnote{Unless stated otherwise, we always exclude zero vectors from the admissible input domain, as any zero argument forces the bilinear output to be the all-zero vector and thus constitutes a degenerate evaluation.}
A $(T, \varepsilon)$-secure private simultaneous messages (PSM) protocol with public information for $\mathcal{F}$ is a triple of PPT algorithms (\alg{Setup}, \alg{Encode}, \alg{Output}) with the following syntax:
\begin{itemize}
	\item \alg{Setup} is randomized and takes as input $\mathbbm{1}^n$ and a description of $f_n$. The output is a triple ($I, s_0, s_1$) where $I$ is public information, and $s_0, s_1$ are private values.
	\item \alg{Encode} is randomized and takes as input $\mathbbm{1}^n$, an index $i \in\{0,1\}$, a public information $I$, a private value $s_i$, and an input $x_i \in \F_q^n$. The output is a message $m_i$.
	\item \alg{Output} is deterministic and takes as input two messages $\hat{x}, \hat{y}$ and public information $I$. The output is $z\in \F_q^m$.
\end{itemize}
We require the following properties:
\begin{itemize}
	\item \textbf{Perfect correctness.} For all inputs $x, y \in\F_q^n$ and random strings $r, r_0, r_1$, if $\left(I, s_0, s_1\right) \smpl$ $\alg{Setup}\left(\mathbbm{1}^n, f_n ; r\right), \hat{x} \smpl \alg{Encode}\left(\mathbbm{1}^n, 0, I, s_0, x ; r_0\right)$, and $\hat{y} \smpl \alg{Encode}\left(\mathbbm{1}^n, 1, I, s_1, y ; r_1\right)$, then the output is correct, i.e., $\alg{Output}\left(\hat{x}, \hat{y}, I\right)=f_n(x, y)$.
	\item \textbf{Security.} There exists a non-uniform polynomial-time simulator $\alg{PSMSim}$ such that, for every non-uniform $(T(n) \cdot \operatorname{poly}(n))$-time adversary $\mathcal{A}$, sufficiently large $n$, and $x, y \in \F_q^n$, 
	
\[
\left|
\begin{aligned}
&\Pr\!\left[
\mathcal{A}(\mathbbm{1}^n,\hat{x},\hat{y},I)=1
\;\middle|\;
\begin{array}{l}
(I,s_0,s_1)\xleftarrow{\$}\alg{Setup}(\mathbbm{1}^n,f_n),\\
\hat{x}\xleftarrow{\$}\alg{Encode}(\mathbbm{1}^n,0,I,s_0,x),\\
\hat{y}\xleftarrow{\$}\alg{Encode}(\mathbbm{1}^n,1,I,s_1,y)
\end{array}
\right]
\\[-1mm]
&\qquad -
\Pr\!\left[
\begin{array}{l}
(\hat{x}',\hat{y}',I')\xleftarrow{\$}
\alg{PSMSim}(\mathbbm{1}^n,f_n(x,y)),\\
\mathcal{A}(\mathbbm{1}^n,\hat{x}',\hat{y}',I')=1
\end{array}
\right]
\end{aligned}
\right|
\leq \varepsilon(n).
\]
	that is, the adversary cannot distinguish with advantage greater than $\varepsilon(n)$ if the messages and public information were generated as in the PSM protocol with inputs $x, y$ or by the simulator only holding the output $f(x, y)$.
\end{itemize}
The message size of a PSM protocol with public information is

$$
\ell(n):=\max _{i, r, r^{\prime}, x \in \F_q^n}\left\{\left|m_i\right| \left\lvert\, \begin{array}{l}
	\left(I, s_0, s_1\right) \leftarrow \alg{Setup}\left(\mathbbm{1}^n ; r\right) \\
	m_i \leftarrow \alg{Encode}\left(\mathbbm{1}^n, i, I, s_i, x ; r^{\prime}\right)
\end{array}\right.\right\} .
$$
\end{definition} 

\subsection{Alternating lifts of bilinear maps}
In this section, we explain how to encode an arbitrary bilinear function into an alternating tensor.
\begin{definition}[Tensor representation of a bilinear function]\label{def:tensor representation}
Let $f:\F_q^r \times \F_q^r \to \F_q^m$ be a bilinear function with $f(x,y)=(x^\tr H_1y,\dots,x^\tr H_my)$, where $x,y\in \F_q^r$ and each $H_k\in \M(r,\F_q)$.
We define the tensor representation of $f$ to be a tensor $H=(H_1,\dots,H_m)\in \T(r\times r\times m,\F_q)$, where the $k$th slice encodes the bilinear form corresponding to the $k$th entry of the output.
\end{definition}
\begin{definition}[Alternating lift]
Let $f:\F_q^r \times \F_q^r \to \F_q^m$ be a bilinear map. An alternating lift of $f$ is a quadruple $\mathcal{L}_f=(d,\tB,U,V)$, where $\tB=(B_1,\dots,B_m)\in \AT(d,m,\F_q), U,V\in \M(d\times r,\F_q)$, such that for every $x,y\in \F_q^r$, $f(x,y)=\bigl((Ux)^\tr B_1 (Vy),\dots,(Ux)^\tr B_m (Vy)\bigr)$.
\end{definition}
In other words, an alternating lift expresses the given bilinear map as the restriction of an alternating bilinear map on a larger space, after two fixed linear embeddings.
\begin{proposition}[Standard alternating lift]\label{prop:alternating_lift}
Every bilinear map $f:\F_q^r \times \F_q^r \to \F_q^m$ admits an alternating lift of dimension $d=2r$.
More precisely, let
\[
B_k=
\begin{bmatrix}
0 & H_k\\
-H_k^\tr & 0
\end{bmatrix}
\in \Lambda(2r,\F_q),
\qquad k\in[m],
\]
and define
\[
U=
\begin{bmatrix}
I_r\\
0
\end{bmatrix},
\qquad
V=
\begin{bmatrix}
0\\
I_r
\end{bmatrix}
\in \F_q^{2r\times r}.
\]
Then $\mathcal{L}_f=(2r,\tB,U,V)$ is an alternating lift of $f$.
\end{proposition}
\begin{proof}
    For every $k\in[m]$ and $x,y\in\F_q^r$, we have 
    \[
(Ux)^\tr B_k (Vy)
=
\begin{bmatrix}
x^\tr & 0
\end{bmatrix}
\begin{bmatrix}
0 & H_k\\
-H_k^\tr & 0
\end{bmatrix}
\begin{bmatrix}
0\\
y
\end{bmatrix}
=
x^\tr H_k y.
\]
Applying this identity coordinate-wise gives $\bigl((Ux)^\tr B_1 (Vy),\dots,(Ux)^\tr B_m (Vy)\bigr)=f(x,y)$ as required.
\end{proof}
From this point onward, the protocol itself depends only on an alternating lift $\mathcal{L}_f=(d,B,U,V)$ and on the planted alternating tensor framework from Section~\ref{sec:assumptions}.

\subsection{PSM protocols with public information from planted tensors}


We now describe the generic protocol. The protocol assumes that we are given an alternating lift $\mathcal{L}_f=(d,\tB,U,V) $ of the target bilinear map $f:\F_q^r \times \F_q^r \to \F_q^m$, where $\tB=(B_1,\dots,B_m)\in \AT(d,m,\F_q)$

To construct the protocol, we planted $\tB$ into a large random tensor $\tA$ to obtain an updated tensor $\tA$. Intuitively, the tensor $\tA$ hides the specific instance $\tB$ within a high-dimensional structure, while the isomorphic transformation ensures that only the parties holding the appropriate information can align their inputs with $\tB$. The public information therefore consists of the obfuscated tensor $\tA$ in which we plant $\tB$. The setup will provide the parties with the partial isomorphic transformation information.




\begin{figure}[H]
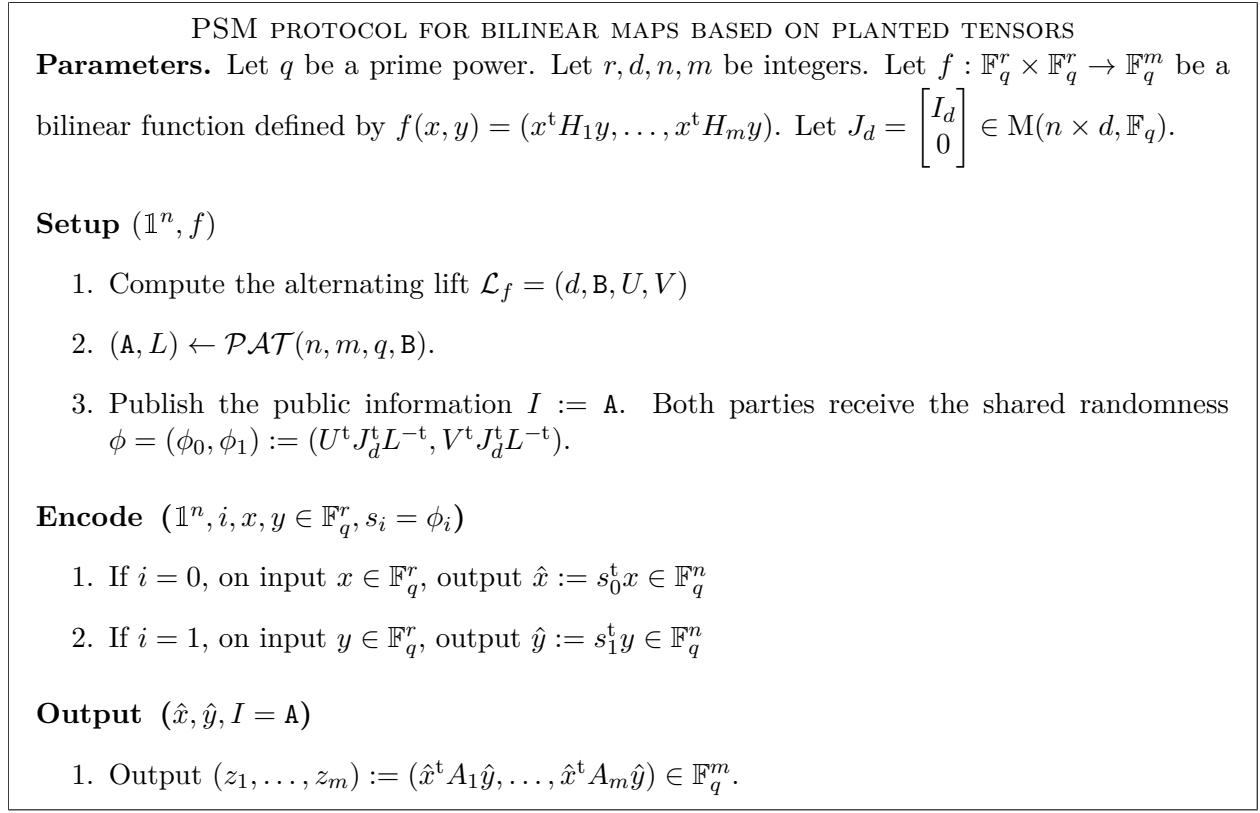

\begin{protocol}
\begin{center}\textsc{PSM protocol for bilinear maps based on planted tensors}\end{center}
\textbf{Parameters.} Let $q$ be a prime power. Let $r,d,n,m$ be integers. Let $f:\mathbb{F}_q^r\times\mathbb{F}_q^r\to\mathbb{F}_q^m$ be a bilinear function defined by $f(x,y)=(x^\tr H_1y,\dots,x^\tr H_my)$.
Let $J_d=\begin{bmatrix}I_d\\0\end{bmatrix}\in\M(n\times d,\F_q)$.
\paragraph{Setup $(\mathbbm{1}^n,f)$}
\begin{enumerate}
\item Compute the alternating lift $\mathcal{L}_f=(d,\tB,U,V)$
	\item $(\tA,L)\gets \mathcal{PAT}(n,m,q,\tB)$.
	\item Publish the public information $I:=\tA$. Both parties receive the shared randomness $\phi=(\phi_0,\phi_1):=(U^\tr J_d^\tr L^{-\tr}, V^\tr J_d^\tr L^{-\tr})$.
\end{enumerate}

\paragraph{Encode \,($\mathbbm{1}^n, i, x,y\in \F_q^r, s_i=\phi_i$)}
\begin{enumerate}
    \item If $i=0$, on input $x\in \F_q^r$, output $\hat{x}:=s_0^\tr x\in \F_q^n$
    \item If $i=1$, on input $y\in \F_q^r$, output $\hat{y}:=s_1^\tr y \in \F_q^n$
\end{enumerate}
\paragraph{Output \,($\hat{x},\hat{y},I=\tA$)}
\begin{enumerate}
	\item Output $(z_1,\dots,z_m):=(\hat{x}^\tr A_1\hat{y},\dots,\hat{x}^\tr A_m\hat{y})\in \F_q^m$.
\end{enumerate}
\end{protocol} 
\caption{PSM protocol for bilinear maps based on planted tensors}
\label{protocol:PSM}
\end{figure}
\begin{theorem}
Consider the PSM protocol constructed in~\cref{protocol:PSM} where Alice and Bob hold inputs $x\in \F_q^r$ and $y\in \F_q^r$ respectively, if the $(n,m,q,T,\varepsilon,\tB)$-PATw2H
assumption holds for the distribution of $\tB$ then the protocol is $(T,\varepsilon)$-secure. Each party sends one vector in $\F_q^n$, so the message size is $n\log q$ bits,
and the public information consists of an $n\times n \times m$ alternating tensor, hence has size $n^2m\log q$ bits. Moreover, assume that Conjecture~\ref{conj:\PRSwH} holds, then, there exists a constant $C\geq3$ such that $d=\lceil n/C\rceil$ and the construction achieves $(T,\varepsilon)$ security with message size of $Cd\log q$ bits and public information size of $C^2d^2
\left\lceil\frac{Cd}{\log(Cd)}\right\rceil\log q
=
O\left(\frac{d^3\log q}{\log d}\right)
$ bits for every adversary running in time $q^{o(n)}$.
\end{theorem}
\begin{proof}
We prove correctness and security separately.\\~
\emph{Correctness.}
Fix any $x,y\in \F_q^r$ and any $k\in [m]$.  Let $A'$ denote the planted tensor before the final hiding step, so that $A=L^\tr A' L$ and $J_d^\tr A'_k J_d = B_k$. By construction of the private values, we have $\hat{x}=s_0^\tr x = L^{-1}J_d Ux$ and $\hat{y}=s_1^\tr y = L^{-1}J_d Vy$. Therefore,
\[
\hat{x}^\tr A_k \hat{y}
=
x^\tr U^\tr J_d^\tr L^{-\tr}(L^\tr A'_k L)L^{-1}J_d V y
=
x^\tr U^\tr J_d^\tr A'_k J_d V y.
\]
Using the identity $J_d^\tr A'_k J_d=B_k$, we obtain
\[
\hat{x}^\tr A_k \hat{y}
=
x^\tr U^\tr B_k V y
=
(Ux)^\tr B_k (Vy).
\]
Since $\mathcal{L}_f=(d,\tB,U,V)$ is an alternating lift of $f$, we conclude that
\[
\bigl(\hat{x}^\tr A_1 \hat{y},\dots,\hat{x}^\tr A_m \hat{y}\bigr)=f(x,y).
\]
This proves perfect correctness.\\
\emph{Security.} 
Fix arbitrary inputs $x,y\in \F_q^r$ and write $x':=Ux\in \F_q^d, y':=Vy\in \F_q^d$. Note that $x'$ and $y'$ are linearly independent by the choice of $U$ and $V$ in \cref{prop:alternating_lift}. Define $z_\tB(x',y'):=\bigl((x')^\tr B_1 y',\dots,(x')^\tr B_m y'\bigr)\in \F_q^m$. By the defining property of the alternating lift, $z_\tB(x',y')=f(x,y)$.

We now describe the adversary's view in the protocol. The public information is the planted alternating tensor $\tA$, and the two messages are $\hat{x}=L^{-1}J_d x'$ and $\hat{y}=L^{-1}J_d y'$, where $(\tA,\tB,L)\xleftarrow{\$} \mathcal{PAT}(n,m,q,\tB)$. Therefore, for the fixed inputs $x,y$, the real protocol view is distributed exactly as $\left(
\tA,\,
L^{-1}J_d x',\,
L^{-1}J_d y'
\right)$.
By Lemma~\ref{lemma:bridge-sim} and PATw2H, there exists a simulator $\alg{Sim}_\alg{PATw2H}$ such that, for every pair of such $x',y'\in \F_q^d$, the real distribution $\left(
\tA,\,
L^{-1}J_d x',\,
L^{-1}J_d y'
\right)$ with $(\tA,\tB,L)\xleftarrow{\$} \mathcal{PAT}(n,m,q,\tB)
$ is computationally indistinguishable from $\alg{Sim}_\alg{PATw2H}(\mathbbm{1}^n,\tB,z_\tB(x',y'))$. Since $z_\tB(x',y')=f(x,y)$, it follows that the real protocol view is computationally indistinguishable from $\alg{Sim}(\mathbbm{1}^n,\tB,f(x,y))$.

We now define the simulator for the PSM protocol by 
$$\alg{PSMSim}(\mathbbm{1}^r, f(x,y)):=\alg{Sim}_\alg{PATw2H}(\mathbbm{1}^n,\tB,z_\tB(x',y'))$$
This simulator depends only on the public description of the functionality through the fixed tensor $\tB$, and on the output value $f(x,y)$; it does not depend on the private
inputs $x$ and $y$ themselves. Consequently, for every non-uniform adversary $\mathcal{A}$ running in time
$T(r)\cdot \poly(r)$, the distinguishing advantage between
\begin{itemize}
    \item  the real protocol transcript generated from $(x,y)$, and
    \item the simulated transcript $\alg{PSMSim}(1^r,f(x,y))$
\end{itemize}
is at most $\varepsilon(r)$.
This is exactly the required security condition for a PSM protocol with public information.
\end{proof}

\begin{remark}
    Although our construction is inspired by \cite{ABIKN23}, the security proof is different from the planted subgraph setting. In the
planted subgraph PSM protocol, the public object encodes a hidden adjacency structure, each party
reveals the planted location of one vertex, and the referee recovers the output by checking a
single adjacency predicate in the public graph. Accordingly, the security argument is centered on
hiding this adjacency information beyond the output bit.

In contrast, our protocol is built around a planted alternating tensor, and the two transmitted
messages are not vertex locations but encoded vectors derived from the hidden congruence
transformation. The referee recovers the output by evaluating a family of bilinear forms $(\hat{x}^\tr A_1 \hat{y},\dots,\hat{x}^\tr A_m \hat{y})$. Our proof no longer reduces to adjacency testing, but instead to preserving only a vector of algebraic relations induced by the planted tensor.
\end{remark}



	
	

\section{Secret Sharing Schemes from Planted (Random) Tensors}\label{sec:secretsharing}
In a $2$-out-of-$d$ secret-sharing scheme (also known as a $(2,d)$-threshold secret-sharing scheme), any two parties can recover the secret but any single party cannot. A secret-sharing scheme for a forbidden graph $Q_d$ generalizes this idea: the parties correspond to the vertices of $Q_d$, and the edges of $Q_d$ show whether the corresponding two parties can recover the secret. More specifically, any two parties whose vertices are connected must be able to recover the secret, while any two vertices without an edge between them are guaranteed to learn nothing about the secret. In analogy to the PSM with public information, this model can be further augmented with some public information that is broadcast to everyone at the beginning. The public information is necessary for recovering the secret by authorized pairs of parties, but it should not give adversaries any advantage in learning about the secret. In this section, we demonstrate how to construct a secret-sharing scheme with public information for forbidden graph access structures under the \PRATwTwoH assumption. 

\subsection{Forbidden graph secret-sharing schemes with public information} 

We now formalize the definitions of forbidden graph access structures and the public information model. 

\begin{definition}[Forbidden graph access structures {\cite{SS97}}]
A $d$-party forbidden graph access structure is a family of sets $(\mathcal{Q}_d, \mathcal{F}_d)_{d \in \mathbb{N}}$ such that
\begin{itemize}
	\item $|S|=2$ for every $S \in \mathcal{Q}_d$.
	\item $\mathcal{Q}_d \cup \mathcal{F}_d=\{S \subseteq[d] \mid |S| \leq 2\}$, $\mathcal Q_d\cap\mathcal F_d=\varnothing$.
\end{itemize}
\end{definition}

This definition follows the generic construction of the access structures in secret-sharing schemes, where $\mathcal{Q}_d$ refers to the set of parties who are able to recover the secret and $\mathcal{F}_d$ contains the set of parties who must not be able to learn anything about the secret. The first condition means that recovering the secret always requires (at least) pairs of parties, while the second condition can then be understood to forbid the remaining pairs of parties, as well as any single party itself, from accessing the secret. So this is naturally encoded as an order-$d$ graph structure, where the $d$ vertices correspond to the $d$ parties, and two vertices are adjacent if and only if they form a $2$-element set $S$ in the authorized set $\mathcal{Q}_d$.

In analogy to the PSM protocols with public information, \cite{ABIKN23} proposed a computational variant of secret-sharing schemes in which the primitive is also augmented with public information. To share a secret value $s$, the dealer generates a collection of small private shares $(s_1, \ldots, s_d)$, assigning one to each participant, along with large public information $I$ that is broadcast to everyone. The public information $I$ is required for successful reconstruction of the secret but, by design, it does not compromise the secrecy of $s$. Working in the computational setting allows $I$ to be used to reduce the size of the private shares. This model is motivated by scenarios where storing or transmitting private information is more costly than handling public data. Furthermore, such schemes often enable secret reconstruction with very limited communication among the parties, a property that is particularly attractive when the same public information can be reused across multiple sharings. Now we slightly adjust the model from a 1-bit secret setting to a $q$-ary string secret, formally defined as follows.

\begin{definition}[Secret-sharing schemes with public information and forbidden graph access structures {\cite{ABIKN23}} for $q$-ary string secret]
Let $m\in\N$. Let $T: \mathbb{N} \rightarrow \mathbb{N}$ be a time bound and let $\varepsilon: \mathbb{N} \rightarrow[0,1]$ be an indistinguishability bound. A $(T, \varepsilon)$-secure secret-sharing scheme with public information for the forbidden graph access structure $(\mathcal{Q}_d, \mathcal{F}_d)_{d \in \mathbb{N}}$ is a pair of uniform PPT algorithms (\textsf{Share}, \textsf{Recover}) with the following syntax:
\begin{itemize}
	\item \textsf{Share} is a randomized algorithm that takes as input $\mathbbm{1}^d$ and a secret $x \in\F_q^m$. Its output is public information $I$ and $d$ strings, called shares, $s_1, \ldots, s_d$, one for each party.
	\item \textsf{Recover} is a deterministic algorithm that takes as input $\mathbbm{1}^d$, a set $S \subseteq[d]$ with $|S|=2$, public information $I$, and shares $(s_i)_{i \in S}$. Its output is either a value $x^{\prime} \in\F_q^m $ or $\perp$.
\end{itemize}
We require the following properties:
\begin{itemize}
	\item \textbf{Perfect correctness.} For every $d \in \mathbb{N}, S \in \mathcal{Q}_d, x \in\F_q^m$, and random string $r$,
	$$
	\text { If } (I, s_1, \ldots, s_d) \leftarrow \textsf{Share}(\mathbbm{1}^d, x; r) \text {, then } \textsf{Recover}(\mathbbm{1}^d, S, I, (s_i)_{i \in S})=x.
	$$
	\item \textbf{Security.} For every non-uniform $(T(d)\cdot\operatorname{poly}(d))$-time adversary $\mathcal{A}$, sufficiently large $d$, and $S \in \mathcal{F}_d$,
	$$
	\left|\Pr\left[\mathcal{A}(\mathbbm{1}^d, I,\left(s_i\right)_{i \in S})=x \left\lvert\, \begin{array}{l}
		x \smpl\F_q^m \\
		\left(I, s_1, \ldots, s_d\right) \leftarrow \textsf{Share}\left(\mathbbm{1}^d, x\right)
	\end{array}\right.\right]-\frac{1}{q^m}\right| \leq \varepsilon(d).
	$$
\end{itemize}
We define the share size of the scheme as
$$
\ell(d):=\max _{i \in[d], x \in\F_q^m, r}\left\{\left|s_i\right| \left|\left(I, s_1, \ldots, s_d\right) \leftarrow \textsf{Share}\left(\mathbbm{1}^d, x ; r\right)\right\}\right..
$$
\end{definition}

This definition mainly follows the modeling in \cite{ABIKN23}. Here, the algorithm \textsf{Share} describes the procedure of broadcasting the public information $I$ and distributing the shares $s_1,s_2,\dots,s_d$ to $d$ parties, one for each; the algorithm \textsf{Recover} requires the authorized set of parties in $\mathcal{Q}_d$ to exactly reconstruct the secret. To achieve the security, adversaries cannot guess the secret value $x\in\F_q^m$ with any non-negligible advantage. It is worth noting that a significant difference between our secret-sharing schemes and the one in \cite{ABIKN23} lies in the secret; we extend their binary secret $b\in\{0,1\}$ to a length-$m$ $q$-ary secret $x \in\F_q^m$. This provides more capacity for the secret information being shared in the schemes.

\subsection{Forbidden graph secret-sharing schemes with public information from planted random tensors}

To ``erase'' the information carried by the public tensor $\tB$ (or equivalently, the bilinear map $\phi_{\tB}$) from the view of forbidden pairs $\{i,j\}$, we need to define a new tensor $\hat{\tB}$ such that $\phi_{\hat{\tB}}(e_i,e_j)=0$, while for authorized pairs $\{i',j'\}$, we still have $\phi_{\hat{\tB}}(e_{i'},e_{j'})=\phi_{\tB}(e_{i'},e_{j'})$. This procedure, visualized as tensors, can be formalized as the following puncturing operation.

\begin{definition}[Tensor puncturing with graph structure]
Let $\tB\in\AT(d, m,\F_q)$ and $Q_d=([d], E)$ be an undirected graph. We define the procedure $\mathsf{Puncture}(\tB, Q_d)$ as follows:
\begin{enumerate}
	\item $\hat{\tB}\leftarrow\tB$.
	\item 
	For every $i,j\in[d]$ such that $i\neq j$, if $i\nsim j$ in $Q_d$, set $\hat{\tB}(i,j,k)=0$ for every $k\in[m]$.
	\item Output $\hat{\tB}$.
\end{enumerate}
\end{definition}
With this puncturing procedure, we are ready to give our forbidden graph secret-sharing schemes with public information based on planted random tensors in \cref{protocol:secret-sharing}.
\begin{figure}[H]
\begin{protocol}
	\begin{center}\textsc{Forbidden graph secret-sharing schemes based on planted random tensors}\end{center}
	Let $Q_d=([d], E)$ be the graph representing the access structure and $x=(x_1,\dots,x_m)\in\F_q^m$ be the secret.\\
	\textsf{Share}$(\mathbbm{1}^d, x)$:
	\begin{enumerate}
		\item $\tB\smpl\mathcal{AT}(d,m,q)$.
		\item $L^{\prime}\smpl\GL(d,\F_q)$
		\item $\widehat{\tB}\leftarrow\mathsf{Puncture}((L^\prime)^\tr \tB L^\prime,Q_d)$.
		\item $(\tilde{\tA},L)\smpl\mathcal{PAT}(n,m,q,\tB)$.
		\item Let $\tilde{\tB}\leftarrow
\mathtt{0}_{d\times d\times m}$. For every edge $\{i,j\}\in E$ with $i<j$ and every $k\in[m]$, set $\tilde{\tB}(i,j,k)\leftarrow x_k+\widehat{\tB}(i,j,k)$. Then set $\tilde{\tB}(j,i,k)\leftarrow -\tilde{\tB}(i,j,k).$
		\item Output $I:=(\tilde{\tA},\tilde{\tB})$ and, for every $i\in[d]$, $s_i:=L^{-1}L^{\prime\prime}e_i$, where $L''=J_dL'\in\mathbb F_q^{n\times d}$.
	\end{enumerate}
	\textsf{Recover}$(\mathbbm{1}^d, S=\{i,j\} (i<j), I, (s_k)_{k\in S})$:
	\begin{enumerate}
		\item If $\{i, j\} \notin E$, output $\perp$.
		\item $(\tilde{\tA},\tilde{\tB})\leftarrow I$.
		\item Denote by $(s_i^\tr\tilde{\tA}s_j)_k$ the $k$th component of $s_i^\tr\tilde{\tA}s_j\in\F_q^m$. Output $\hat{x}=(\hat{x}_1,\dots,\hat{x}_m)$ where $$\hat{x}_k=\tilde{\tB}(i,j,k)-(s_i^\tr\tilde{\tA}s_j)_k$$ for each $k\in[m]$.
	\end{enumerate}
\end{protocol}
\caption{Forbidden graph secret-sharing schemes based on planted random tensors}
\label{protocol:secret-sharing}
\end{figure}

\begin{theorem}
If the $(n, m, d, q, T, \varepsilon)$-\PRATwTwoH assumption holds, the construction in \cref{protocol:secret-sharing} is a $(T, 2 \varepsilon)$-secure forbidden graph secret-sharing scheme with share size of $n\log q$ and public information size of $(n^2+d^2) m\log q$. Moreover, assume that Conjecture \ref{conj:\PRATwTwoH} holds, then there exists a constant $C\geq 3$ such that $d=\lceil n/C\rceil$ and the construction achieves $(T, 2\varepsilon)$ security with share size $C\cdot d\log q$ and public information size $(C^2+1)d^2
\left\lceil\frac{Cd}{\log(Cd)}\right\rceil\log q
=
O\left(\frac{d^3\log q}{\log d}\right)$ for every adversary running in time $T=q^{o(n)}$.
\end{theorem}
\begin{proof}
We check the required properties one by one:
\begin{itemize}
    \item \textbf{Perfect correctness.} For every pair $\{i,j\}\in E$ with $i<j$ that is allowed to recover the secret $x=(x_1,\dots,x_m)\in\F_q^m$, 
    $$
    s_i^\tr\tilde{\tA}s_j = e_i^\tr (L^{\prime\prime})^\tr L^{-\tr}\tilde{\tA} L^{-1}L^{\prime\prime}e_j=e_i^\tr (L^{\prime})^\tr\tB L^{\prime}e_j = e_i^\tr \tilde{\tB} e_j - x.
    $$

    According to the scheme in \cref{protocol:secret-sharing}, for each $k\in[m]$, they output
    $$
    \hat{x}_k=\tilde{\tB}(i,j,k)-(s_i^\tr\tilde{\tA}s_j)_k=\tilde{\tB}(i,j,k)-\tilde{\tB}(i,j,k)+x_k=x_k.
    $$
    \item \textbf{Security.} We prove it by contradiction. Suppose that there exists a non-uniform $T(d)\cdot\mathrm{poly}(d)$-time adversary $\mathcal{A}_1$ that breaks the security of the secret-sharing scheme in \cref{protocol:secret-sharing}. Then there exists a subsequence $(d_k)_{k\in\N}$ and sets $(S_d)_{d\in\N}$ such that $S_d \in \mathcal{F}_d$, $|S_d|\leq 2$ and 
	$$
	\Pr\left[\mathcal{A}_1(\mathbbm{1}^{d_k}, I,\left(s_i\right)_{i \in S_{d_k}})=x \left\lvert\, \begin{array}{l}
		x \smpl\F_q^m \\
		\left(I, s_1, \ldots, s_{d_k}\right) \leftarrow \textsf{Share}\left(\mathbbm{1}^{d_k}, x\right)
	\end{array}\right.\right] > \frac{1}{q^m} + 2\varepsilon(d_k).
	$$
    for every $k\in\N$. Now we use $\mathcal{A}_1$ to attack the \PRATwTwoH assumption. Given an instance of the underlying problem for the \PRATwTwoH assumption, we assume $\mathcal{A}_2$ gets $\tA,\tB$ and $(s_i)_{i\in S_d}$ and performs as follows:
    \begin{enumerate}
        \item If $d\neq d_k$ for every $k\in \N$, output $0$.
        \item Sample $x\smpl\F_q^m$.
        \item Sample $L^\prime\smpl\GL(d,\F_q)$.
	\item $\widehat{\tB}\leftarrow\mathsf{Puncture}((L^\prime)^\tr \tB L^\prime,Q_d)$.
    \item Let $\tilde{\tB}\leftarrow
\mathtt{0}_{d\times d\times m}$. For every edge $\{i,j\}\in E$ with $i<j$ and every $k\in[m]$, set $\tilde{\tB}(i,j,k)\leftarrow x_k+\widehat{\tB}(i,j,k)$. Then set $\tilde{\tB}(j,i,k)\leftarrow -\tilde{\tB}(i,j,k).$
	\item Run $\mathcal{A}_1$ on input $I:=(\tA,\tilde{\tB})$ and $(s_i)_{i\in S_d}$.
        \item If $\mathcal{A}_1$ guesses $x$, then output $1$; otherwise, output $0$.
    \end{enumerate}
    Note that if $\tB$ is not planted in $\tA$, which means $\tA$ is a random tensor without any information that can be leveraged for comparison with $\tilde{\tB}$, then $\mathcal{A}_1$ cannot recover the secret $x$. If $\tB$ is indeed planted in $\tA$, then $\mathcal{A}_1$ will be exactly dealing with the secret-sharing game in \cref{protocol:secret-sharing}. Since $\mathcal{A}_1$ can recover $x$ with probability greater than $1/q^m + 2\varepsilon(d_k)$ for every $d=d_k$ in this game, the advantage of $\mathcal{A}_2$ against the \PRATwTwoH assumption is greater than $\varepsilon(d_k)$ for every $d=d_k$. The total running time of $\mathcal{A}_2$ is $T(d)\cdot\mathrm{poly}(d)$, which contradicts \cref{conj:\PRATwTwoH} and thereby proves that such $\mathcal{A}_1$ does not exist under this conjecture. \qedhere
\end{itemize}
\end{proof}

\bibliographystyle{alpha}
\bibliography{group-action}

\end{document}

%% file: preamble.tex
\usepackage[T1]{fontenc}
\usepackage{graphicx}
\usepackage[normalem]{ulem}
\usepackage{hyperref}
\usepackage{tikz}
\usetikzlibrary{calc,decorations.pathreplacing}
\usepackage{color}

\usepackage{amsthm}
\usepackage{caption}
\usepackage{algorithmicx, algpseudocode, algorithm}
\usepackage{amsmath,amsfonts,amssymb}
\usepackage{todonotes}
\usepackage[operators,sets]{cryptocode}
\usepackage{expl3}
\usepackage{soul}
\usepackage{xspace}
\usepackage{xcolor}

\usepackage{enumitem}
\usepackage[capitalise]{cleveref}
\usepackage[left=1in,right=1in,bottom=1in,top=1in]{geometry}
\usepackage{float}
\usepackage{tabularx}
\usepackage{mdframed}
\usepackage{mathtools}
\setul{1ex}{.5pt}

\newenvironment{protocol}{
	\begin{mdframed}[style=figstyle]}{
\end{mdframed}}

\newtheorem{theorem}{Theorem}
\newtheorem{fact}[theorem]{Fact}
\newtheorem{claim}[theorem]{Claim}
\newtheorem{proposition}[theorem]{Proposition}
\newtheorem{definition}{Definition}
\newtheorem{lemma}[theorem]{Lemma}
\theoremstyle{definition}

\newtheorem{remark}{Remark}
\newtheorem{conjecture}{Conjecture}

\newcommand{\bzero}{\mathbf{0}}
\newcommand{\F}{\mathbb{F}}

\newcommand{\N}{\mathbb{N}}

\newcommand{\T}{\mathrm{T}}
\newcommand{\AT}{\mathrm{AT}}
\newcommand{\vA}{\mathbf{A}}
\newcommand{\vB}{\mathbf{B}}

\newcommand{\tA}{\mathtt{A}}
\newcommand{\tB}{\mathtt{B}}

\newcommand{\M}{\mathrm{M}}
\newcommand{\Gr}{\mathrm{Gr}}

\newcommand{\ER}{\mathcal{G}}
\newcommand{\LinER}{\mathcal{AT}}
\newcommand{\PlLinER}{\mathcal{AT}}
\newcommand{\PPlLinER}{\mathcal{PAT}}

\newcommand{\rank}{\mathrm{rank}}
\newcommand{\linspan}{\mathrm{span}}

\newcommand{\tr}{\mathrm{t}}

\newcommand{\alg}[1]{\textsf{\upshape #1}}

\newcommand{\PRATwTwoH}{PlRanAltTensorW2Hints}
\newcommand{\PRSwH}{PlRanSubgraphWHints}

\DeclareMathOperator{\GL}{\mathrm{GL}}

\DeclareMathOperator{\poly}{\mathrm{poly}}

\DeclareMathOperator*{\E}{\mathbf{E}}
\DeclareMathOperator*{\Var}{\mathbf{Var}}

\usepackage{tikz}
\usetikzlibrary{matrix}
\newif\ifstudent
\newcommand{\equ}[1]{%
	\[
	\ifstudent\else\expandafter\phantom\fi{#1}
	\]%
}
\usepackage[most]{tcolorbox}
\usepackage{lmodern}
\newtcbox{\transparentsetting}{blank, on line, opacitytext=0.15}

\newcommand{\smpl}{\xleftarrow{\$}}

\usepackage[numbers,sort]{natbib}
\usepackage{bbm}